%% file: main.tex
\documentclass[12pt,fleqn,a4paper]{article}
\usepackage{amsmath}
\usepackage{amssymb}
\usepackage{amsfonts}
\usepackage{textcomp}
\usepackage{amsthm}
\usepackage{xspace}
\usepackage{fullpage}
\usepackage[english]{babel}

\usepackage{authblk}
\usepackage[utf8]{inputenc}
\usepackage{tikz}
\usetikzlibrary{patterns}
\usetikzlibrary{snakes}
\usetikzlibrary{calc}
\usepackage{subcaption}

\theoremstyle{definition}
\newtheorem{theorem}{Theorem}[section]
\newtheorem{lemma}{Lemma}[section]
\newtheorem{corollary}{Corollary}[section]
\newtheorem{proposition}{Proposition}[section]

\newtheorem{definition}{Definition}[section]

\newcommand\vhdef2

\newcommand{\class}[1]{{\ifnum\vhdef=2\mathbf{#1}\else\mathrm{#1}\fi}}
\newcommand{\co}{
\ifnum\vhdef=2\mathbf{co\,}\else\mathrm{co\hspace{2pt}}\fi
\ifnum\vhdef=2\textbf{-}\else\textrm{-}\fi
}

\newcommand{\lang}[1]{\mathtt{#1}}

\newcommand{\LEAFD}{\lang{LEAFD}}
\newcommand{\RLEAFD}{\lang{RLEAFD}}
\newcommand{\ARROWS}{\lang{ARROWS}}

\newcommand{\BROUWERD}{\lang{2D-BROUWER}}

\newcommand{\FNP}{\class{FNP}}
\newcommand{\TFNP}{\class{TFNP}}
\newcommand{\PPAD}{\class{PPAD}}

\newcommand{\NP}{\class{NP}}

\newcommand{\comment}[1]{}

\newcommand{\CSP}{{\rm CSP}}
 \newcommand{\al}{\leftarrow}
\newcommand{\alu}{\nwarrow}
\newcommand{\au}{\uparrow}
\newcommand{\aru}{\nearrow}
\newcommand{\ar}{\rightarrow}
\newcommand{\ard}{\searrow}
\newcommand{\ad}{\downarrow}
\newcommand{\ald}{\swarrow}

\usepackage{url}
\title{Tree-like resolution complexity of two planar problems}

\author[1]{Dmitry Itsykson\thanks{dmitrits@pdmi.ras.ru, partially supported by the RFBR grant 14-01-00545, 
by the President's grant MK-2813.2014.1,
by the Government of the Russia (grant 14.Z50.31.0030) and by the RAS program of fundamental research}}
\author[2]{Anna Malova\thanks{malova.anya@gmail.com}} 
\author[2]{Vsevolod Oparin\thanks{oparin.vsevolod@gmail.com, partially supported by ANR NAFIT 008-01}}
\author[1]{Dmitry Sokolov\thanks{sokolov.dmt@gmail.com, partially supported by the RFBR grant 12-01-31239-mol-a, 
by the President's grant MK-2813.2014.1,
by the Government of the Russia (grant 14.Z50.31.0030)}}
\affil[1]{St. Petersburg Department of
V.A. Steklov Institute of Mathematics of
the Russian Academy of Sciences
}

\affil[2]{St. Petersburg Academic University of the Russian Academy of Sciences
}

\begin{document}

\maketitle

\input{parts/abstract}

\input{parts/intro}
\input{parts/prelim}
\input{parts/lower}
\input{parts/sperner}
\input{parts/arrows}

\input{parts/ppad}

\section*{Acknowledgements}
The authors are grateful to Alexander Shen for fruitfull discussions and the statement of the problem and also thank 
Mikhail Slabodkin for helpfull comments.

\bibliography{main}
\bibliographystyle{plain}

\newpage\appendix
\input{parts/appendix}

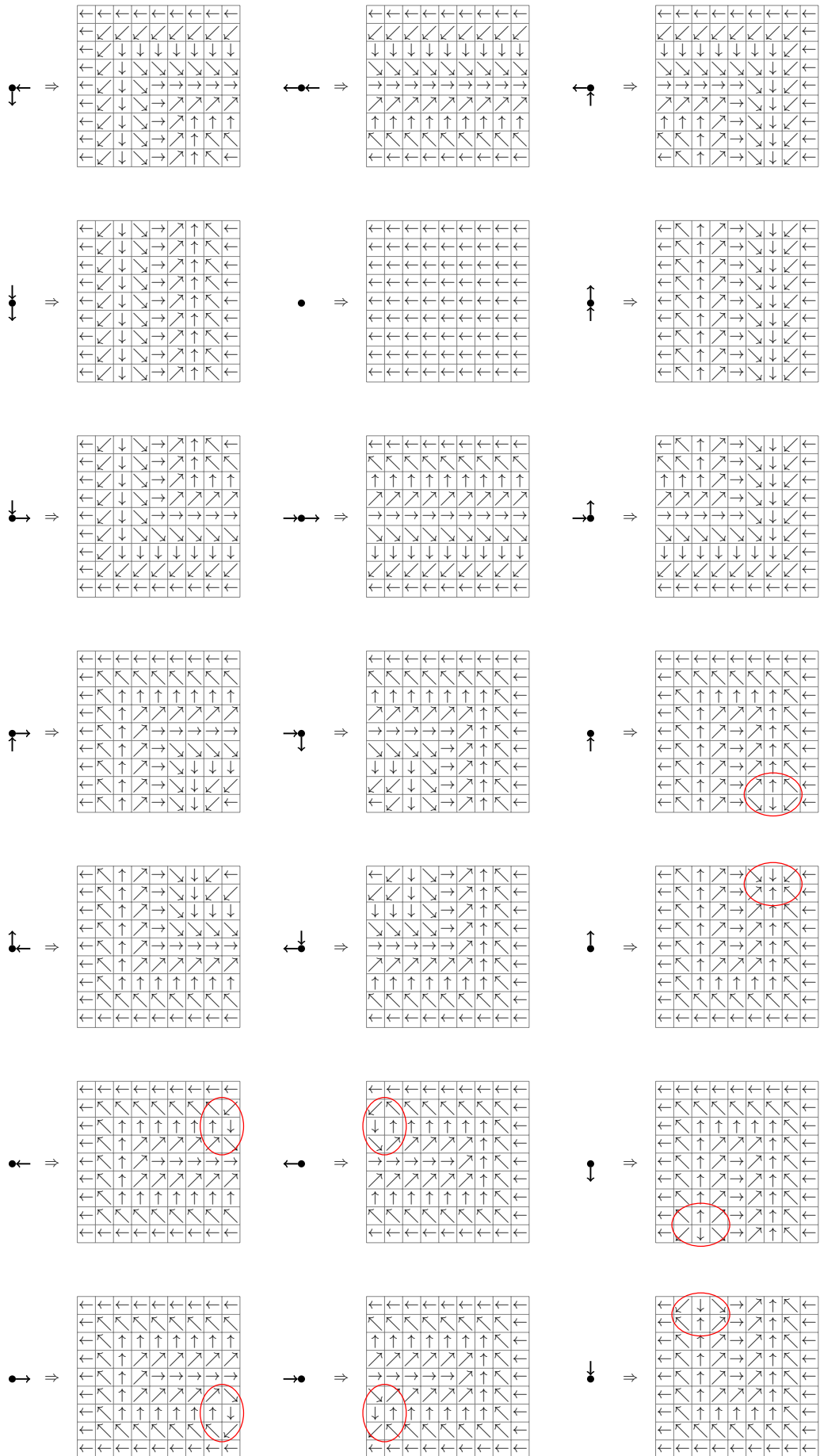
\end{document}

%% file: parts/abstract.tex
\begin{abstract}
 
We consider two CSP problems: the first CSP encodes 2D Sperner's lemma for the standard triangulation of the right triangle
on $n^2$ small triangles; the second CSP encodes the fact that it is  impossible to 
match cells of $n \times n$ square to arrows (two horizontal,  two vertical and four diagonal) such that arrows in two cells  
with a common edge differ by at most $45^\circ$, and all arrows on  the boundary of the square do not look outside 
(this fact is a corollary of the Brower's fixed point theorem). 
We prove that the tree-like resolution complexities of these CSPs are $2^{\Theta(n)}$. For Sperner's lemma our result implies
$\Omega(n)$ lower bound on the number of request to colors of vertices that is enough to
make  in order to find a trichromatic triangle; this lower bound was originally proved by
Crescenzi and Silvestri.

CSP based on Sperner's lemma is related with the $\rm PPAD$-complete problem. We show that CSP corresponding to arrows is also 
related with a $\rm PPAD$-complete problem. 

\end{abstract}

%% file: parts/intro.tex
\section{Introduction}
\label{intro}

The resolution proof system is one of the simplest and well-studied proof systems. There are well known methods of proving lower and upper bounds on the complexity of several types of formulas. However there are no known universal methods that may be used to determine the asymptotic resolution complexity of a given family of formulas.


Baker \cite{Baker95} extended resolution proof system for constraint satisfaction problems
($\CSP$) under arbitrary alphabets. The resolution proof system is connected with backtracking algorithms (so called DPLL algorithms). Every DPLL algorithm for  $\CSP$ under $k$-element alphabet  creates $k$-ary tree. Every node of this tree corresponds to a variable, and edges, going to children, correspond to  $k$ different substitutions of the variable. Every leaf of the tree contains a constraint  that is falsified by substitution made on the path from the root to that leaf. We call such tree as contradiction search tree. It is well known that the minimal size of 
a contradiction search tree for an unsatisfiable $\CSP$ is equal to the size of the minimal tree-like resolution  proof of $\CSP$.

Every unsatisfiable constraint satisfaction problem $F$ under $k$-element alphabet has three parameters: $S(F)$ is the minimal size of resolution proof of $F$,  $S_T(F)$ is the minimal size of tree-like resolution proof, and $d(F)$ is the minimal depth of contradiction search tree. 
These parameters are connected by trivial inequalities: $S(F) \le S_T(F) \le k^{d(F)}$. 
The paper~\cite{BEGJ00} presents the family of formulas that exponentially separate $S$ and $S_T$, and  paper~\cite{Urq11} gives an example of family $F_n$ such that $S_T(F_n) = O(n)$ and $d(F_n) = \Omega(n / \log n)$. The number $d(F)$ is the worst-case lower bound on the number of requests to the variables of $\CSP$ that is necessary to do in order to find a falsified constraint.

We consider the constraint satisfaction problem that codes 2D Sperner lemma: 
namely is impossible to color vertices of the standard triangulation of the right triangle in three colors
such that vertices of the right triangle are colored in different colors and all vertices on the side are 
colored in two colors and there are no small triangles cored with three different colors. We prove that that size of 
tree-like refutation of this CSP is at least $2^{\Omega(n)}$, where $n$ is the number of points on the side of the triangle.
Previous result by Crescenzi and R. Silvestri \cite{CS98} gives a linear lower bound on the depth of resolution proof for the same CSP.

We also consider the constraint satisfaction problem that codes that it is impossible to match cells of the $n \times n$ square with arrows from the set $\{\al, \alu, \au, \aru, \ar, \ard, \ad, \ald\}$ such that  arrows in every two adjacent cells differ by at most $45^\circ$ and boundary arrows are not directed outside the square. The impossibility is a corollary of the Brower's fixed point theorem. We prove that tree-like resolution complexity of this $\CSP$ is equal to $2^{\Theta(n)}$. 

Our research is motivated by the investigation of complexity classes $\rm PPA, PPAD, PPADS, PPP$ that are subclasses of  
$\rm TFNP$ function problems that are guaranteed to have a solution because of unsatisfiability of certain $\CSP$ \cite{Pap94}.  
It is known that Sperner's Lemma corresponds to $\rm PPAD$-complete problem. 
We show that $\rm CSP$ with arrows also corresponds to $\rm PPAD$-complete problem. 




%% file: parts/prelim.tex
\section{Preliminaries}

Let $X = \{x_1, x_2, \dots, x_n\}$ be a finite number of variables that can take values from a finite alphabet
$W = \{w_1, w_2, \dots, w_k\}$.
Let $S$ be the set of constraints for $X$: every constraint depends on some subset $X' \subseteq X$ and defines a set of values 
that variables from $X'$ can take simultaneously. 
A constraint satisfaction problem ($\CSP$) is a triplet $\langle X, W, S \rangle$. 
A constraint satisfaction problem is satisfiable if
there is a set of values for $X$ that satisfies all constraints from $S$, otherwise we call $\CSP$ unsatisfiable.

{\em A partial substitution} is a map from $X$ to $W \cup \{*\}$. We say that substitution $\rho$ sets value for variable $x$ if $\rho(x) \neq *$.
A substitution is {\em full} if values of all variables are set.

Substition $\rho$ falsifies constraint $C \in S$ if values of all variables, that constraint $C$ depends on, are defined by $\rho$ and 
constraint $C$ forbids such a setting of values by $\rho$.

{\em A contradiction search tree} for an unsatisfiable $\CSP$ is a rooted $k$-arity tree such that vertices are marked with variables 
and edges, that connect a vertex marked with $x$ to its children, are marked with substitutions $x := w$ for every $w \in W$. 
Every leaf of the tree is marked with a constraint that is refuted by the substitution made along the path from the root to the leaf.

A nogood is a constraint of the following form: $\neg (x_1 = a_1 \wedge \cdots \wedge x_\ell = a_\ell)$, where
$x_1, \dots, x_\ell \in X, a_1, \dots, a_\ell \in W$. In the case of binary alphabet $W = \{0, 1\}$  nogoods are equivalent to clauses. 
If constraint depends on $\ell$ variables, then it can be represented as a conjunction of at most $|W|^\ell$ nogoods.

The resolution proof system can be generalized to $\CSP$ if all its constraints are represented as conjunctions of nogoods.
A resolution proof for some $\CSP$ formula $\phi$ is a sequence of nogoods $C_1, C_2, \dots, C_t$, that ends with empty nogood $\Box$. 
Every nogood is either contained in $\phi$ or can be derived from $k$ previous nogoods by the resolution rule. 
Let $\{N_a\}_{a\in W}$ be the set of nogoods in the following form: $N_a = \neg(x = a \wedge \alpha_a)$ for all $a \in W$. 
Then nogood $\neg (\bigwedge_{a \in W}\alpha_a)$ is a resolvent (a result of resolution rule) of nogoods $\{N_a\}_{a \in W}$. 
A resolution proof is called tree-like if there exists a $k$-ary tree such that it leaves are marked with 
nogoods of initial formula $\phi$ and nogood in any other vertex $v$ is a resolvent of nogoods 
that are written in children of $v$; the root of the tree is marked with the empty nogood $\Box$.

\begin{proposition}[\cite{Baker95}\cite{Hwang04}] 
The size of the minimal tree-like resolution proof for every unsatisfiable $\CSP$ formula $\phi$
is equal to the size of minimal contradiction search tree for $\phi$.
\end{proposition}

%% file: parts/lower.tex
\section{Tree size lower bounds}

Consider $\CSP$ $\phi$ under a finite alphabet $W = \{w_1, w_2, \dots, w_k\}$ that 
depends on variables $X = \{x_1, x_2, \dots, x_n\}$ and consists of constraints $C_1,
C_2, \dots, C_m$.

We consider the following game that will be used for proving lower bounds on the size of contradiction search trees.
A game is defined by an unsatisfiable $\CSP$ formula $\phi$; two players Alice and Bob play as follows.
Alice secretly from Bob chooses a contradiction search tree for $\phi$ and put a token in the root of the tree.
Alice asks Bob about the value of a variable that corresponds a vertex that contains the token.
Bob either returns a value of asked variable or suggest to Alice choose the value by herself from some 
subset of $W$ of cardinality at least 2. In the second case we say that Bob moves ChooseAny. Alice moves the token 
according the Bob's answer, if Bob moves ChooseAny, then Alice chooses a value herself from the set proposed by Bob. 
The game is over if the token is in a leaf, that is a contradiction is found.
The goal of Bob is to maximize the number of ChooseAny answers along the path from the root to a leaf.

\begin{lemma}
\label{decision-tree-lemma}
	Let for $\CSP$ formula $\phi$ there exist such a strategy of Bob that for all strategies of Alice,
    Bob moves ChooseAny at least $t$ times. Then the size of any contradiction search tree for $\phi$ is at least $2^t$.
\end{lemma}

\begin{proof} 
Consider some contradiction search tree $T$ for $\phi$. 
We construct probabilistic distribution on the leaves of $T$ that corresponds to Bob's strategy and
the following randomized strategy of Alice. Alice chooses the tree $T$ and if Bob moves ChooseAny, Alice 
chooses a value at random with equal probabilities. By the statement of the Lemma the probability 
that the game will finish in every particular leaf is at most $2^{-t}$. Since with probability $1$ 
the game will finish in a leaf, the number of leafs is at least $2^t$.
\end{proof}

In applications of Lemma~\ref{decision-tree-lemma} it is convenient to describe the strategy of Bob in the following terms: 
Bob has a partial substitution that assigns values to variables that Alice asked for and probably to some other variables.
The current substitution should not falsify constraints. If Alice asks for a variable that is determined by the substitution,
Bob returns the assigned value, otherwise Bob moves ChooseAny and extends the substitution by the value that Alice chooses and probably
substitutes values to some other variables.

%% file: parts/sperner.tex
\subsection{Sperner's lemma}
\label{sectsperner}

\newcommand{\com}[1]{{\color{blue}(#1)}}
\newcommand{\smart}[1]{{\color{red}(#1)}}


Let us consider triangle $T$, each side of whom is divided on $n - 1$ segments of the same length.
Denote the ends of the segments by $p_1, \dots, p_n$, $p_n = q_1 \dots, q_n$, $q_n = r_1, \dots, r_n = p_1$ as it is shown on Figure~\ref{sp-invariants}.

The standard triangulation is obtained by joining the following pairs of vertices:
\begin{enumerate}
	\item $p_i$ and $q_{n - (i-1)}$, $\forall i (2 \le i \le n - 1)$
	\item $p_i$ and $r_{n - (i-1)}$, $\forall i (2 \le i \le n - 1)$
	\item $q_i$ and $r_{n - (i-1)}$, $\forall i (2 \le i \le n - 1)$
\end{enumerate}
The crossings between the above straight-line segments are called crossing vertices.

Denote all vertices in the triangle by $S$. A vertex coloring $c:~S~\to~\{Red, Green, Blue\}$ 
is said to be Sperner's coloring if vertices $p_1$, $q_1$, and $r_1$ are labeled by three different 
colors and each vertex on each edge of triangle $T$ is colored by only 
one of two colors of the ends of this edge.

In the triangulation we will discuss only small triangles, i.e. triangles without crossings inside.
A triangle in the triangulation is said to be {\em trichromatic} if all its vertices are colored with three different colors.

\begin{lemma}[Sperner~\cite{Spe80}]
\label{sperner-lemma}
If $c$ is a Sperner coloring, the triangulation contains a trichromatic triangle.
\end{lemma}

Let denote the unsatisfable CSP corresponding to Lemma~\ref{sperner-lemma} by $\phi_n$.
Variables correspond to the vertices of the triangulation and can take values from alpabet $\{Red, Green, Blue\}$. 
Three constraints fix particular colors for vertices of triangle $T$ and we set a constraint for each vertex on each side of $T$, which fixes a set of two possible colors.
Also for each small triangle in the triangulation we set a constraint which forbids its trichromatic coloring.
The unsatisfability of $\phi_n$ immediately follows from Sperner's lemma.

\begin{lemma}
\label{bob-sperner-lemma}
	For $\CSP$ $\phi_n$ there exists such a strategy of Bob that gives
   $\Omega(n)$ ChooseAny answers for any strategy of Alice.
\end{lemma}

Before the proof we associate variables of the $\CSP$ with vertices of the triangulation and a partial substitution with a coloring.
Some vertices can be non-colored, so values of the corresponding vertices has not been set.

\begin{proof}

We provide the strategy of Bob and then prove the lower bound.

As it is described in the previous section, Bob maintains a coloring. 
If Alice asks Bob for a colored vertex, Bob returns its color.

Bob maintains a path $P$ which starts from $r_{\lceil \frac{n}{2}\rceil}$ and stops at a crossing vertex $f$. 
Path $P$ goes along edges of the triangulation and can use only right or right-down directions.
Bob keeps in mind two segments which start from vertex $f$ and go in right and right-down directions to the border of the triangle.
We denote them by $RP$ and $DP$, respectively.
We denote the parallelogram, surrounded by $RP$, $DP$ and the border of triangle $T$, by $H$

{\em Buffer} is an area that is strictly in parallelogram $H$ and contains all points whose distance to $RP$ or $DP$ are at most two.
Note that the buffer does not include points on $RP$ and $DP$. 

Informally Bob hides a trichromatic triangle in parallelogram $H$. Formally he maintains the following invariants (see Figure~\ref{sp-invariants}):

\begin{enumerate}
	\item\label{inv-bord} The border vertices are colored according the following
	\begin{enumerate}
	\item $c(p_1) = c(r_i) = Blue$, for any $i$, with $\lceil \frac{n}{2} \rceil - 1 \le i \le n - 1$
	\item $c(q_n) = c(r_i) = Green$, for any $i$, with $2 \le i \le \lceil \frac{n}{2} \rceil$
	\item $c(p_i) = c(q_j) = Red$, for any $i$, $j$, with $2 \le i \le n$, $2 \le j \le n - 1$
	\end{enumerate}	
	\item\label{inv-buf} Vertices in the buffer are empty. It means Alice has not asked for them.
	\item\label{inv-path} Any vertex on path $P$ is green, and any vertex just below the path is blue.
	\item\label{inv-p-rp} Any vertex above $P \cup RP$ is colored, and each blue vertex is surrounded by red vertices.
	\item\label{inv-p-dp} Any vertex left to $P \cup DP$ is colored, and each green vertex is surrounded by red vertices.
	\item\label{inv-rp-dp} Any vertex on $RP$ is green, any vertex on $DP$ is blue.
	\item\label{inv-paral} Any green or blue vertex in $H$, excluding $RP$ and $DP$, is surrounded by red vertices.
	\item\label{inv-empty} 
All uncolored vertices are in $H$.
\end{enumerate}

\begin{figure}
\centering
\input{pics/sp-invariant2.tex}
\caption{Invariants} 
\label{sp-invariants}
\end{figure}
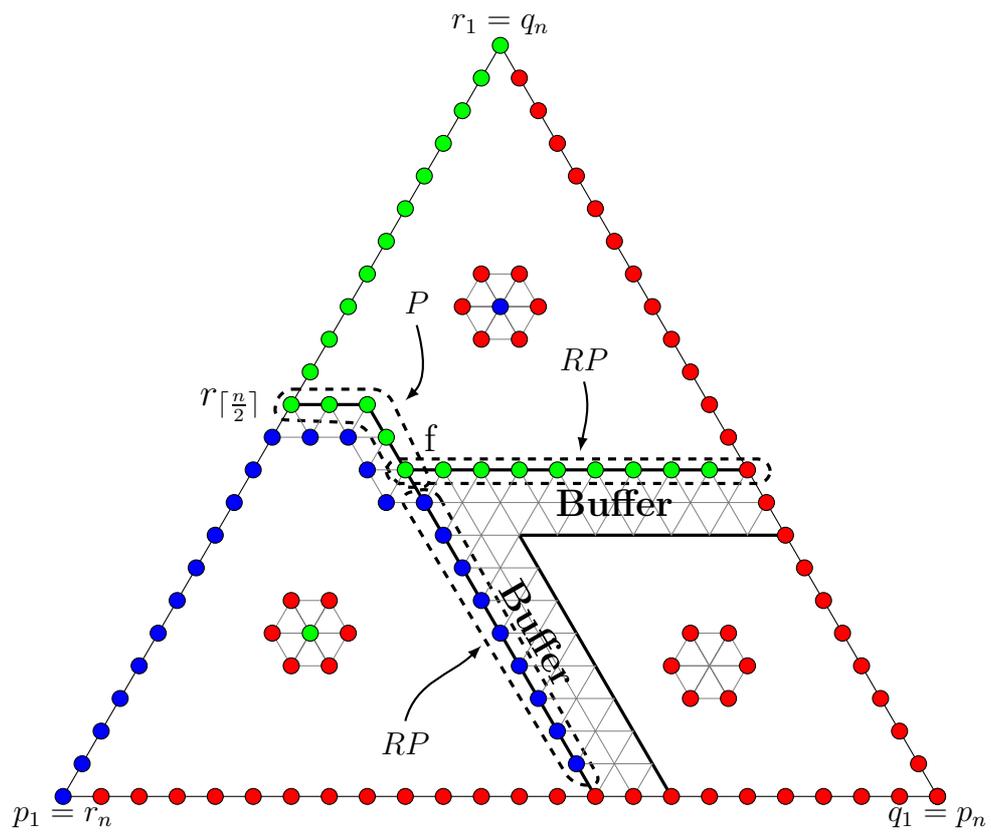

It is easy to see that if all invariants are hold, then there are no trichromatic triangle in the current coloring. 

Absence of such a triangle above $P \cup RP$ (including border $P \cup RP$) follows from Invariant~\ref{inv-p-rp}.
There is not adjacent vertices colored on blue and green. Similarly absence of a trichromatic triangle left to $P \cup DP$ 
follows from Invariant~\ref{inv-p-dp}. Invariant~\ref{inv-path} guarantees that vertices on $P$ can not lie on a trichromatic triangle, because all these vertices don't have adjacent red vertex below. 
Finally Invariants~\ref{inv-buf} and \ref{inv-paral} implies that there are no trichromatic triangles in $H$.

We will prove that if after a query of Alice, Bob cannot 
support these invariants, he has alredy given $\Omega(n)$ ChooseAny 
answers.

Initially $f = r_{\lceil \frac{n}{2} \rceil}$, $P$ contains only $f$. Segments $RP$ and $DP$ are defined by the current position of vertex $f$. In the current coloring all vertices in area above $RP$ are colored by green and ones left to $DP$ by 
blue.

Let Alice ask for a non-colored vertex. By Invariant~\ref{inv-empty} such a vertex lies in parallelogram $H$. 
Bob answers accordingly the following cases, which are shown on 
Figure~\ref{cases}. 

\begin{figure}
\centering
\input{pics/sp-cases.tex}
\caption{Cases} 
\label{cases}
\end{figure}

In Case~4 Bob returns ChooseAny(Green, Blue) and surrounds 
selected vertex by red ones to hold Invariant~\ref{inv-paral}.

In cases 1~--~3 Bob has to move the buffer to hold Invariant~\ref{inv-buf}. 
In order to explain how to move teh buffer, we introduce two new notions: a slot and an expanded slot. 
The parallelogram $H$ is a union of horizontal and diagonal segments. The crossings of these segments are vertices of the triangulation.
For the horizontal segments we skip first three ones on the top and divide all others onto the groups with four segments in each.
We do the same thing with the vertical segments skipping three leftmost ones.
We ignore the last segments of each type if their number is at most 3.

{\em An expanded slot} is a union of two groups of horizontal and diagonal groups of segments.

For an expanded slot we consider the parallelogram which is formed by the borders 
of triangle $T$, the next to the leftmost segment in the vertical group and the next to top segment in horizontal one.
{\em A slot} is a subset of the considered expanded slot which also lies in the described parallelogram.

The left and top borders of a slot can be used for new $RP$ and $DP$, respectively. The rest of the slot
can be used for the new buffer.

We say that a group of segments is clear if Alice has not requested any vertex in this group.

If Alice's request fits Cases 1~--~3, Bob looks for the first clear 
horizontal and diagonal groups of segments starting from vertex $f$, forms an expanded slot 
and extends path $P$ as it will be described below. The new end of path $P$ defines new position of parallelogram $H$ and segments $RP$ and $DP$.Bob adds all vertices, which lie outside of new parallelogram $H$ or belong segment $RP$ and 
$DP$, in the partial substitution as it is described in 
Invariants~\ref{inv-p-rp} and ~\ref{inv-p-dp}.

Now we are ready to describe how Bob should answer on Alice's requests.

\begin{enumerate}
	\item Bob answers ChooseAny(Green, Blue). If Alice chooses green, Bob 
	extends $P$ accordinally Figure~\ref{ext-2}, 
	and Figure~\ref{ext-1} otherwise. 
	\item Bob answers ChooseAny(Green, Red) and extends $P$ as showed on 
	Figure~\ref{ext-2} regardless Alice's choice.
	\item Bob answers ChooseAny(Blue, Red) and extends $P$ as showed on 
	Figure~\ref{ext-1} regardless Alice's choice.
	\item Bob answers ChooseAny(Green, Blue) and colors all uncolored
	vertices adjacent to requested vertex by red.
\end{enumerate}

Note that all of these assignments do not violate the invariants. Path $P$ 
is extended by empty vertices or colored by green and 
all adjacent vertices lies below $P$ are empty or colored by blue.  
The fourth case guarantees satisfying Invariant~\ref{inv-paral}, and \ref{inv-p-rp} and
\ref{inv-p-dp} in the future.

\begin{figure}
\centering
\input{pics/sp-ext-1_new.tex}
\caption{The First Extension of path $P$} 
\label{ext-1}
\end{figure}
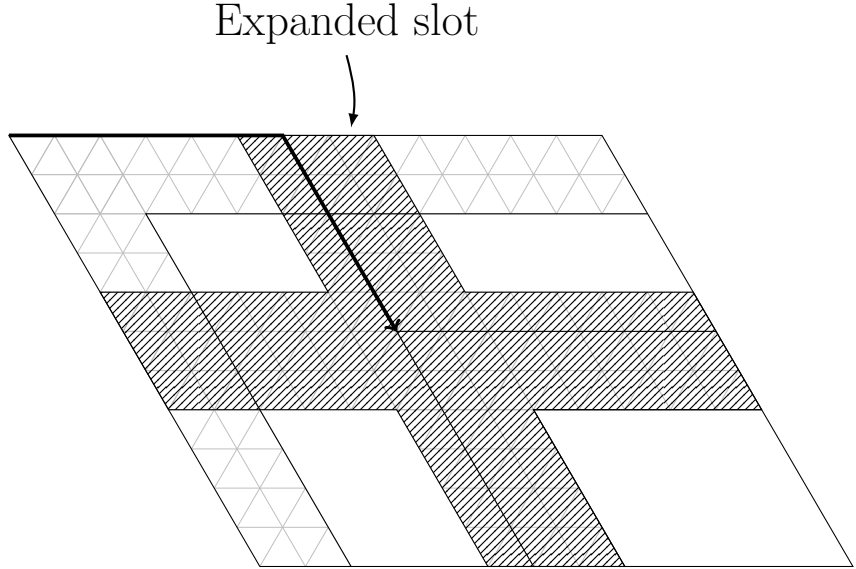 

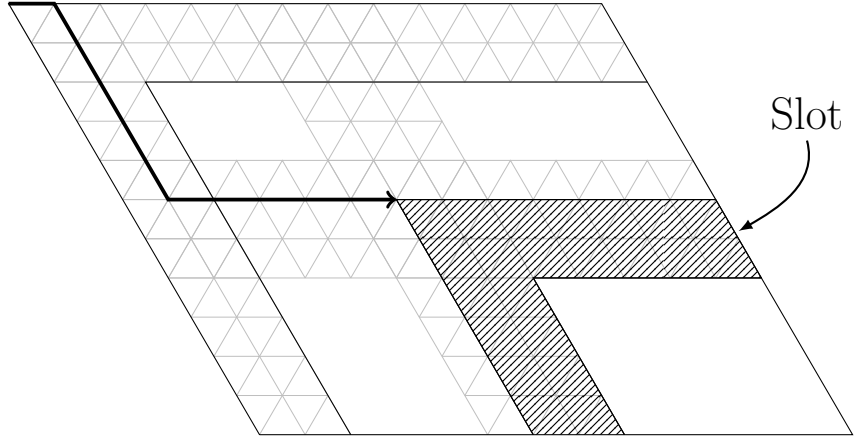
\begin{figure}
\centering
\input{pics/sp-ext-2_new.tex}
\caption{The Second Extension of path $P$} 
\label{ext-2}
\end{figure} 


At the start of the game the number of clear groups of segments in both 
directions is $\Omega(n)$. In a turn Alice can spoil at most two groups in each 
direction. Every time Alice spoils a group, Bob either gives ChooseAny answer 
or finishes the game. The latter means that Bob can not hold all invariants. 

Every time Bob moves the buffer he skips only spoiled groups. So at the 
end of the game all slots in at least one direction will be spoiled. 
The latter guarantees that Bob gives $\Omega(n)$ ChooseAny answers.

Remind that Bob's strategy guarantees that a trichromatic triangle can not be found until
all invariants are hold. So for any strategy of Alice Bob gives $\Omega(n)$ ChooseAny answers. 

\end{proof}

%% file: pics/sp-invariant2.tex
\begin{tikzpicture}[scale=.5]
  \def\n{24}	
  \def\m{12} 
  \def\k{23} 

    \draw [dashed, very thick, rounded corners = 5pt]
    	(5.7, 10.8) -- (8.5, 10.8) -- (9.75, 8.2) -- (8.8, 8.2) --
	    (7.9, 9.9) -- (5.5, 10) -- cycle;
	\draw[-latex, thick] (9.3, 12.5) node[above]
    	{$P$} to[out = -75, in = 45] (9, 10.5);
    
	\draw [dashed, very thick, rounded corners = 5pt]
    	(8.5, 8.95) -- (18.6, 8.95) -- (18.6, 8.3) -- (8.5, 8.3) -- cycle;
	\draw[-latex, thick] (13.7, 11) node[above]
        {$RP$} to[out = -75, in = 80] (13.6, 9.15);
  
	\draw [dashed, very thick, rounded corners = 5pt]
    	(9.75, 8.2) -- (14.2, 0.3) -- (13.4, 0.3) -- (8.8, 8) -- cycle;
	\draw[-latex, thick] (9, 2) node[below]
        {$RP$} to[out = 75, in = -135] (11, 4);
  
  \draw(0,0) -- ++ (0:\n-1);
  \draw(0,0) -- ++ (60:\n-1);
  \draw(0,0) ++ (0:\n-1) -- ++ (120:\n-1);
  
  \draw[very thick] (0,0) ++(60:\m) -- ++(0:2) -- ++ (-60:2);
  \draw[style=help lines] (0,0) ++(60:\m-1) -- ++(0:3);
  \draw[style=help lines] (0,0) ++(60:\m) -- ++(-60:1)-- ++(60:1) -- ++ (-60:1) -- ++ (60:1);
  \draw[style=help lines] (0,0) ++(60:\m-1) ++(0:2) -- ++(-60:1) -- ++(60:1);
  \draw[style=help lines] (0,0) ++(60:\m-2) ++(0:3) -- ++(0:1);
  
  \draw[very thick] (0,0) ++(60:\m) ++(0:3) ++(-60:2) ++ (180:1) -- ++(0:9);
  \draw[style=help lines] (0,0) ++(60:\m) ++(0:2) ++(-60:3) -- ++(0:9);
  \draw[very thick] (0,0) ++(60:\m) ++(0:4) ++(-60:4) -- ++(0:7);
  
  \draw[very thick] (0,0) ++(60:\m) ++(0:2) ++(-60:2) -- ++(-60:10);
  \draw[style=help lines] (0,0) ++(60:\m) ++(0:2) ++(-60:4) ++ (0:1) -- ++(-60:8);
  \draw[very thick] (0,0) ++(60:\m) ++(0:2) ++(-60:4) ++ (0:2) -- ++(-60:8);

  \draw[style=help lines] (0,0) ++(60:\m) ++(0:2) ++(-60:3) -- ++(60:1) -- ++(-60:1) -- ++(60:1) -- ++(-60:1) -- ++(60:1) -- ++(-60:1) 
                                          -- ++(60:1) -- ++(-60:1) -- ++(60:1) -- ++(-60:1) -- ++(60:1) -- ++(-60:1) 
                                          -- ++(60:1) -- ++(-60:1) -- ++(60:1) -- ++(-60:1) -- ++(60:1); 
  \draw[style=help lines] (0,0) ++(60:\m) ++(0:2) ++(-60:4) -- ++(60:1) -- ++(-60:1) -- ++(60:1) -- ++(-60:1) -- ++(60:1) -- ++(-60:1) 
                                          -- ++(60:1) -- ++(-60:1) -- ++(60:1) -- ++(-60:1) -- ++(60:1) -- ++(-60:1) 
                                          -- ++(60:1) -- ++(-60:1) -- ++(60:1) -- ++(-60:1) -- ++(60:1);                                           
  
  \draw[style=help lines] (0,0) ++(60:\m) ++(0:2) ++(-60:4) -- ++(0:1) -- ++(-120:1) -- ++(0:1) -- ++(-120:1) -- ++(0:1) -- ++(-120:1)
										  -- ++(0:1) -- ++(-120:1) -- ++(0:1) -- ++(-120:1) -- ++(0:1) -- ++(-120:1)
										  -- ++(0:1) -- ++(-120:1) -- ++(0:1) -- ++(-120:1);
										  
  \draw[style=help lines] (0,0) ++(60:\m) ++(0:3) ++(-60:4) -- ++(0:1) -- ++(-120:1) -- ++(0:1) -- ++(-120:1) -- ++(0:1) -- ++(-120:1)
										  -- ++(0:1) -- ++(-120:1) -- ++(0:1) -- ++(-120:1) -- ++(0:1) -- ++(-120:1)
										  -- ++(0:1) -- ++(-120:1) -- ++(0:1) -- ++(-120:1);
	
  \draw[style=help lines] (0,0) ++(60:\m-1) ++(0:2) ++(-60:1) -- ++(-60:1) -- ++(0:1) ++(120:1) -- ++(-120:1);
  \foreach \point in {1,...,\m} {
	\filldraw[fill=blue] (0,0) ++(60:\point-1) circle(6pt);
	\filldraw[fill=green] (0,0) ++(60:\m+\point-1) circle(6pt);
  }
  
  \foreach \point in {1,...,\k} {
	\filldraw[fill=red] (0,0) ++(0:\point) circle(6pt);
	\filldraw[fill=red] (0,0) ++(0:\n-1) ++(120:\point-1) circle(6pt);
  }
  
  \filldraw[fill=green] (0,0) ++(60:\m) ++(0:1) circle(6pt) ++(0:1) circle(6pt) ++(-60:1) circle(6pt) ++ (-60:1) circle(6pt);
  \filldraw[fill=blue] (0,0) ++(60:\m-1) ++(0:1) circle(6pt) ++(0:1) circle(6pt) ++(-60:1) circle(6pt) ++(-60:1) circle(6pt);

  \foreach \point in {1,...,8} {
	\filldraw[fill=green] (0,0) ++(60:\m) ++(0:2) ++(-60:2) ++(0:\point) circle(6pt);
  }
  
  \foreach \point in {1,...,9} {
	\filldraw[fill=blue] (0,0) ++(60:\m) ++(0:2) ++(-60:2) ++(-60:\point) circle(6pt);
  }
  
  \draw[style=help lines] (60:5) ++(0:4) -- ++(180:1) -- ++(60:1) -- ++(-60:1)
					   -- ++(120:1) -- ++(0:1) -- ++(-120:1)
					   -- ++(60:1) -- ++(-60:1) -- ++(180:1)
					   -- ++(0:1) -- ++(-120:1) -- ++(120:1)
					   -- ++(-60:1) -- ++(180:1) -- ++(60:1)
					   -- ++(-120:1) -- ++(120:1) -- ++(0:1);
					   
  \filldraw[fill=red] (60:5) ++(0:3) circle(6pt) ++(60:1) circle(6pt) ++(0:1) circle(6pt) ++(-60:1) circle(6pt) ++(-120:1) circle(6pt) ++(180:1) circle(6pt);
  \filldraw[fill=green] (60:5) ++(0:4) circle(6pt);	
	
  \draw[style=help lines] (60:\m+3) ++(0:4) -- ++(180:1) -- ++(60:1) -- ++(-60:1)
					   -- ++(120:1) -- ++(0:1) -- ++(-120:1)
					   -- ++(60:1) -- ++(-60:1) -- ++(180:1)
					   -- ++(0:1) -- ++(-120:1) -- ++(120:1)
					   -- ++(-60:1) -- ++(180:1) -- ++(60:1)
					   -- ++(-120:1) -- ++(120:1) -- ++(0:1);
					   
  \filldraw[fill=red] (60:\m+3) ++(0:3) circle(6pt) ++(60:1) circle(6pt) ++(0:1) circle(6pt) ++(-60:1) circle(6pt) ++(-120:1) circle(6pt) ++(180:1) circle(6pt);
  \filldraw[fill=blue] (60:\m+3) ++(0:4) circle(6pt);
  
  \draw[style=help lines] (60:4) ++(0:15) -- ++(180:1) -- ++(60:1) -- ++(-60:1)
					   -- ++(120:1) -- ++(0:1) -- ++(-120:1)
					   -- ++(60:1) -- ++(-60:1) -- ++(180:1)
					   -- ++(0:1) -- ++(-120:1) -- ++(120:1)
					   -- ++(-60:1) -- ++(180:1) -- ++(60:1)
					   -- ++(-120:1) -- ++(120:1) -- ++(0:1);
					   
  \filldraw[fill=red] (60:4) ++(0:14) circle(6pt) ++(60:1) circle(6pt) ++(0:1) circle(6pt) ++(-60:1) circle(6pt) ++(-120:1) circle(6pt) ++(180:1) circle(6pt);


  \draw (60:12) node[left = 7] {\large $r_{\lceil \frac{n}{2} \rceil}$};  
  \draw (60:\m-3)++(0:10) node[rotate=0] {\bf \large Buffer};
  \draw (60:\m-2)++(0:4) node[rotate=0, above right = 3] {\large f};
  \draw (60:\m-7)++(0:10) node[rotate=-60] {\bf \large Buffer};
  \draw (0,0) node[below] {$p_1 = r_n$};
  \draw (0:\n-1) node[below] {$q_1 = p_n$};
  \draw (60:\n-1) node[above] {$r_1 = q_n$};
\end{tikzpicture}

%% file: pics/sp-cases.tex
\begin{tikzpicture}
  \def\n{6}	
  \def\m{8}

  \draw (0,0) -- ++(120:\n) -- ++(0:\m);
  \draw (0, 0)-- (0:\m) -- ++(120:\n);
  
  \draw[style=help lines] (0:3) -- ++(120:\n);
  \draw[style=help lines] (120:\n-3) -- ++(0:\m);
	
  \foreach \level in {0,...,2} {
	  \foreach \point in {1,...,\m} {
		\draw[style=help lines] (120:\n-\level) ++ (0:\point) -- ++(-120:1) -- ++(120:1);
	  }
	  
	  \foreach \point in {1,...,\n} {
		\draw[style=help lines] (0:\level) ++ (120:\point-1) -- ++ (60:1) -- ++(180:1);
	  }
  }
  
  \foreach \level in {1,...,2} {
	\draw[very thick] (0:\level) -- ++(120:\n-\level);
	\draw[very thick] (0,0) ++(0:\level) ++(120:\n-\level) -- ++(0:\m-\level);
  }
  
  \node at ($(120:\n-1) + (0:1)$) [circle,draw=black,fill=white] {1};
  
  \foreach \point in {3,...,\m} {
	\node at ($(120:\n-1) + (0:\point-1)$) [circle,draw=black,fill=white] {2};
	\node at ($(120:\n-2) + (0:\point-1)$) [circle,draw=black,fill=white] {2};
  }
  
  \foreach \point in {4,...,\n} {
	\node at ($(0:1) + (120:\point-3)$) [circle,draw=black,fill=white] {3};
	\node at ($(0:2) + (120:\point-3)$) [circle,draw=black,fill=white] {3};
  }
  \node at ($(120:\n-2) + (0:1)$) [circle,draw=black,fill=white] {3};
  
  \foreach \point in {1,...,3} {
	 \node at ($(0:3) + (120:\point)$) [circle,draw=black,fill=white] {3};
  }

  \foreach \point in {4,...,\m} {
	\node at ($(120:3) + (0:\point-1)$) [circle,draw=black,fill=white] {2};
  }
  
  \foreach \point in {0,...,\m} {
	\node at ($(0,0) + (120:\n) + (0:\point)$) [circle, fill=green] {};
  }
  
  \foreach \point in {1,...,\n} {
	\node at ($(0,0) + (120:\point-1)$) [circle, fill=blue] {};
  }
  
  \foreach \point in {4,...,7} {
   \node at ($(120:1) + (0:\point)$) [circle,draw=black,fill=white!80!gray] {4};
   \node at ($(120:2) + (0:\point)$) [circle,draw=black,fill=white!80!gray] {4};
  }

  \node at ($(0,0) + (120:\n)$) [circle, fill=green] {\large f};
  \draw (120:\n)++(0:4) node[rotate=0, above right] {\large RP};
  \draw (120:\n-4) node[rotate=-60, below left] {\large DP};
  
\end{tikzpicture}

%% file: pics/sp-ext-1_new.tex
\begin{tikzpicture}[scale=.6]
	\tikzstyle{hlines}=[color=gray!50!white, line width=0.1mm]
  \def\n{11}	
  \def\m{13}
  
  \foreach \level in {1,...,2} {
	\draw[hlines] (0:\level) -- ++(120:\n);
	\draw[hlines] (120:\n-\level) -- ++(0:\m);
  }
  
  \foreach \level in {0,...,1} {
	  \foreach \point in {1,...,\m} {
		\draw[hlines] (120:\n-\level) ++ (0:\point) -- ++(-120:1) -- ++(120:1);
	  }
	  kz
	  \foreach \point in {1,...,\n} {
		\draw[hlines] (0:\level) ++ (120:\point-1) -- ++ (60:1) -- ++(180:1);
	  }
  }
  
  \foreach \level in {0,...,3} {
	\draw[hlines] (0,0) ++(120:\n-4-\level) -- ++(0:\m);
	
	\draw[hlines] (0,0) ++(0:5+\level) -- ++(120:\n);
  }
  
  \foreach \level in {0,...,2} {
	\foreach \point in {1,...,\m} {
		\draw[hlines] (120:\n-5-\level) ++ (0:\point-1) -- ++(60:1) -- ++(-60:1);
	}
	
	\foreach \point in {1,...,\n} {
		\draw[hlines] (0,0) ++(0:5+\level) ++ (120:\point-1) -- ++(60:1) -- ++(180:1); 
	}
  }
  
  \draw (0,0) -- ++(120:\n) -- ++(0:\m) -- ++(-60:\n) -- ++(180:\m);
  \draw (0,0) ++(0:2) -- ++(120:\n-2) -- ++(0:\m-2);
  
  \draw (0,0) ++(0:6) -- ++(120:\n-5)-- ++(0:\m-6);
  
  \draw (0,0) ++(0:8) -- ++(120:\n-7) -- ++(0:\m-8);

  \draw[pattern = north east lines] (120:\n) ++ (0:5) -- ++(0:3) --
  	  ++(-60:4) -- ++(0:5) -- ++(-60:3) -- ++(180:5) -- ++(-60:4) --
  	  ++(180:3) -- ++(120:4) -- ++(180:5) -- ++(120:3) -- ++(0:5) -- cycle;

  \draw[-latex, thick] (1.9, 11.3) node[above]
      {\Large Expanded slot} to[out = -75, in = 80] (2, 9.7);  
	
  
  \draw[->, line width=0.5mm] (120:\n) -- ++(0:6) -- ++(-60:5);
\end{tikzpicture}

%% file: pics/sp-ext-2_new.tex
\begin{tikzpicture}[scale=.6]
	\tikzstyle{hlines}=[color=gray!50!white, line width=0.1mm]
  \def\n{11}	
  \def\m{13}

  \foreach \level in {1,...,2} {
	\draw[hlines] (0:\level) -- ++(120:\n);
	\draw[hlines] (120:\n-\level) -- ++(0:\m);
  }
  
  \foreach \level in {0,...,1} {
	  \foreach \point in {1,...,\m} {
		\draw[hlines] (120:\n-\level) ++ (0:\point) -- ++(-120:1) -- ++(120:1);
	  }
	  kz
	  \foreach \point in {1,...,\n} {
		\draw[hlines] (0:\level) ++ (120:\point-1) -- ++ (60:1) -- ++(180:1);
	  }
  }
  
  \foreach \level in {0,...,3} {
	\draw[hlines] (0,0) ++(120:\n-4-\level) -- ++(0:\m);
	
	\draw[hlines] (0,0) ++(0:5+\level) -- ++(120:\n);
  }
  
  \foreach \level in {0,...,2} {
	\foreach \point in {1,...,\m} {
		\draw[hlines] (120:\n-5-\level) ++ (0:\point-1) -- ++(60:1) -- ++(-60:1);
	}
	
	\foreach \point in {1,...,\n} {
		\draw[hlines] (0,0) ++(0:5+\level) ++ (120:\point-1) -- ++(60:1) -- ++(180:1); 
	}
  }
  
  \draw (0,0) -- ++(120:\n) -- ++(0:\m) -- ++(-60:\n) -- ++(180:\m);
  \draw (0,0) ++(0:2) -- ++(120:\n-2) -- ++(0:\m-2);

  \draw[pattern = north east lines] (0,0) ++(0:6) -- ++(120:\n - 5) --
      ++(0:\m-6) -- ++(300:\n - 9) -- ++(180:\m - 8) -- ++(300:\n - 7);

  \draw[-latex, thick] (12, 6.5) node[above]
      {\Large Slot} to[out = -75, in = 30] (10.5, 4.5);  
  
  \draw (0,0) ++(0:6) -- ++(120:\n-5)-- ++(0:\m-6);
  \draw (0,0) ++(0:8) -- ++(120:\n-7) -- ++(0:\m-8);
	
  
  \draw[->, line width=0.5mm] (120:\n) -- ++(0:1) -- ++(-60:5) -- ++(0:5);
\end{tikzpicture}

%% file: parts/arrows.tex
\subsection{Arrows \label{sectarrows}}

We consider one more $\CSP$, denote it as $\psi_n$, that corresponds to the game on a $n\times n$ square. 

The variables correspond to the cells of the square and can take values from set
$W = \{\al, \alu, \au, \aru, \ar, \ard, \ad, \ald\}$.

There are two types of constraints.
1) Two arrows in two cells with a common edge differ by at most $45^\circ$.
2) All arrows on the boundary of the $n \times n$ square should not be directed outside of the square.
We set a constraint for each pair of adjacent cells and each cell on the border.

The unsatisfiability of $\psi_n$ follows from Brouwer's theorem.

\begin{theorem}
\label{arrow-impossible}
    $\CSP$ $\psi_n$ is unsatisfiable. 
\end{theorem}

\begin{corollary} 
\label{arrowsapper}
    For the $\CSP$ $\psi_n$ there exists a contradiction search tree of depth $O(n)$ and therefore of size $2^{O(n)}$.
\end{corollary}

For proofs of Theorem~\ref{arrow-impossible} and Corollary~\ref{arrowsapper} see Appendix~\ref{appendixarrow}.

\begin{theorem}
\label{arrow-thm}
	The size of any contradiction search tree for $\CSP$ $\psi_n$ is
    $2^{\Omega(n)}$, and therefore the depth is $\Omega(n)$.
\end{theorem}

The proof of Theorem~\ref{arrow-thm} follows from Lemma~\ref{decision-tree-lemma} and following one.

\begin{lemma}
\label{bob-arrow-lemma}
	For $\CSP$ $\psi_n$ there exists such a strategy of Bob that Bob gives
    $\Omega(n)$ ChooseAny answers for any strategy of Alice.
\end{lemma}

\begin{proof}
Starting from the left side Bob mentally split the square on vertical strips
 $n \times 8$ that we will call \textit{slots}. Since $n$ may not be divisible by $8$, 
 there may be a strip of size at most $n \times 7$, which remains unsplited.

Let us divide a slot into two strips of size $n \times 4$. We say that left one is a buffer position.
At the beginning of the game Bob divides the rightmost slot into to strips $n \times 4$ and chooses left one.
We call that strip as a buffer. Bob fills all cells, which lie to the right of the buffer, by $\al$.

During the game Bob may move the buffer, but the new position for the buffer is determined in the same way.
Bob divides a slot in two parts and chooses left one.

    
Bob maintains the following invariants:
1) The buffer is empty. 
2) The area, to the right of the buffer, is filled completely. 
The leftmost strip $n \times 1$ to the right of the buffer is filled with $\al$.
3) All neighbours of empty cells are either empty or filled by $\ar$. 
4) Let $k$ be the number of slots to the right of the buffer. Then Bob gave at least $\frac{k}{2}$ answers ChooseAny 
for cells to the right of the buffer.
5) No constraints are falsified.

Note that all invariants are hold at the begining.

The Bob's strategy has the following type:
1) If Alice requests already filled cell, then Bob returns its value.
2) If Alice requests empty cell, then Bob moves ChooseAny and probably fill some cells in order to maintain invariant.
3) If it is impossible to maintain invariants, Bob moves in arbitrary way.

We specify the second step and prove that if Bob fails to maintain invariants, then he has already made $\Omega(n)$ ChooseAny answers. 
    
Let Alice request an empty cell $c$. We consider three cases:
1) the cell $c$ is to the left of the buffer and does not share edges with cells from the buffer;
2) the cell $c$ is in the buffer;
3) the cell $c$ is in the $n \times 1$ strip that is adjacent to and to the left of the buffer.

In the first case Bob answers ChooseAny($\ar$, $\ard$) if $c$ is on the upper boundary and ChooseAny($\ar$, $\aru$), otherwise.

Bob assigns $\ar$ to empty neighbors of	$c$. It is easy to see that such assignments do not violate invariants.

In the second case $c$ lies in the buffer. We consider two cases
a) there are no empty slots left of the buffer;
b) there is an empty slot left of the buffer.

    a) Assume that there is no empty slot left of the buffer. 
Every of $t$ slots left of the buffer has a filled cell. Since for any request to an empty cell left of the buffer 
Bob returns ChooseAny and 
assigns values for no more than two slots, Bob has already made at least $\frac{t}{2}$ ChooseAny answers. Also Bob has made
$\frac{k}{2}$ ChooseAny answers for cells to the right of the buffer, where $k$ is the number of slots to the right of the buffer.
Finally, Bob has made $\frac{t + k}{2} = \Omega(n)$ ChooseAny answers.

  b) If there is an empty slot left of the buffer, Bob chooses the rightmost one and denote it by $S$. 
  He divides the slot into two strips $n \times 4$ and chooses left one for new buffer $B$.
  We denote the old buffer by $B'$.  
We denote the area between $S$ and $B'$ by $M$. Bob is looking for an empty horizontal strip $H$ of height $8$ in $M$, 
that is not adjacent to lower and upper boundary.
Note that if there is no such a strip, then Bob has made $\Omega(n)$ answers ChooseAny for cells in $M$.
 
Bob assigns $\ar$ to all empty cells from $M \setminus H$ and
fills the parts $S \setminus B$ and $M\cap H$ as it is shown in Figure~\ref{pic:middle-arrow}.

We split the old buffer $B'$ into three parts: above, on, and below strip $H$. 
Bob assings cells from $B'$ in two ways: with Pattern~$1$ or Pattern~$2$, shown on Figure~\ref{pic:middle-arrow}. 
Note that these two patterns has the following properties: there is no cell in strip $H$, which is filled with the same arrow in both patterns.
So Bob can answer ChooseAny suggesting two variants according Patterns $1$ or $2$.

    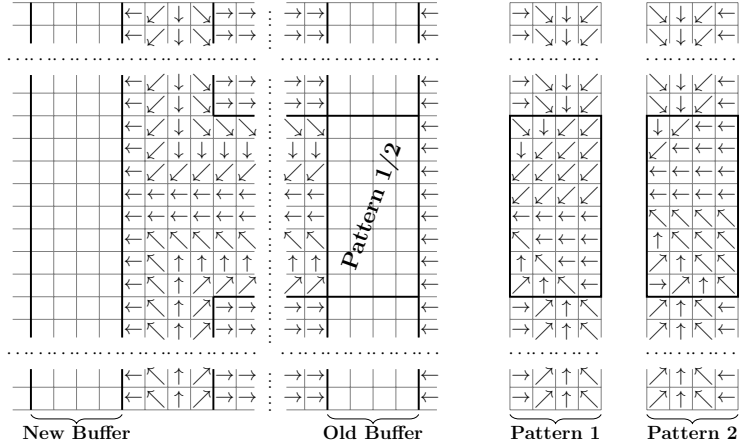
\begin{figure}[h]
		\center{\scalebox{0.6}{\input{pics/a-full.tex}}}
		\caption{Assignments between buffers, patterns for old buffer \label{pic:middle-arrow}}
    \end{figure}

It is easy to see that all invariants are hold and constraints are not violated after such assignments. 
Let $k$ be the number of slots that are to the right of buffer $B'$ and $s$ is the number of slots between old and new buffers.
There are $k + s + 1$ slots to the right of buffer $B$. 
Bob made at least  $\frac{k}{2}$ ChooseAny answers in the cells that are to the right of buffer $B'$, 
all $s$ slots between $B$ and $B'$ are nonempty. 
The first request to every of those slots got the answer ChooseAny 
and after recovering of the invariants at most two slots may be affected. Thus we have at least $\frac{s}{2}$
answers ChooseAny for these $s$ slots. Finally, the number of ChooseAny answers to the right of buffer $B'$ is at least $\frac{s}{2} + \frac{k}{2} + 1 >
    \frac{s + k + 1}{2}$.
    
In the third case $c$ is a left neighbour of the buffer. In this case Bob gives answers ChooseAny($\ar$, $\ard$) as it was in the first case. Bob has to fill the neighbours of cell $c$ by $\ar$, but some of them lie in the buffer.
Thus we move the buffer as in the second case, but fill it with Pattern~1 (see Figure~\ref{pic:middle-arrow}).
So we do not violate the constraints.
The number of ChooseAny answers are estimated as follows: $\frac{k}{2}$ answers to the right of $B'$ and at least
  $\frac{s + 1}{2}$ others (in contrast to the second case, there is no ChooseAny in $B'$ but Bob gave the ChooseAny answer in at least half of $s+1$ remaining slots). 
Thus there are at least $\frac{s + k + 1}{2}$ ChooseAny answers to the right of the new buffer $B$.
\end{proof}

%% file: pics/a-full.tex
\vspace{0.3cm}
    \tikzstyle{arr}=[xshift=7, yshift=7]
    \newcommand{\darr}[3]{\node[arr] at (#2 / 2 - 0.5, #1 / 2 - 0.5) {$#3$}}
    \newcommand{\darrl}[5]{\darr{#1}{-1}{#2};\darr{#1}{0}{#3};\darr{#1}{1}{#4};\darr{#1}{2}{#5};}
    \newcommand{\darrs}[5]{\darr{#1}{-10}{#2};\darr{#1}{-9}{#3};\darr{#1}{-8}{#4};\darr{#1}{-7}{#5};}

    \begin{tikzpicture}[scale=1]
        \begin{scope}
            \draw[step=.5cm,style=help lines] (-1.9,-4.5) grid (1.4,-3.6);
            \draw[step=.5cm,style=help lines] (-1.9,-2.9) grid (1.4, 2.9);
            \draw[step=.5cm,style=help lines] (-1.9,3.6) grid (1.4, 4.5);
            
            \draw[very thick] (-1,-2.9) -- (-1,2.9);
            \draw[very thick] (-1,-4.5) -- (-1,-3.6);
            \draw[very thick] (-1,4.5) -- (-1,3.6);
            \draw[very thick] (1,-2.9) -- (1,2.9);
            \draw[very thick] (1,-4.5) -- (1,-3.6);
            \draw[very thick] (1,4.5) -- (1,3.6);
            \draw[very thick] (-1.9,2) -- (1,2);
            \draw[very thick] (-1.9,-2) -- (1,-2);

            % ...
            \foreach \i in {-15,...,-5,-3,-2,...,3} {
                \darr{7}{\i}{\cdots};
                \darr{-6}{\i}{\cdots};
            }
            \foreach \i in {-8,-7,-5,-4,...,6,8,9} {
                \darr{\i}{-4}{\vdots};
            }

            \draw[step=.5cm,style=help lines] (-7.9,-4.5) grid (-2.6,-3.6);
            \draw[step=.5cm,style=help lines] (-7.9,-2.9) grid (-2.6, 2.9);
            \draw[step=.5cm,style=help lines] (-7.9,3.6) grid (-2.6, 4.5);
            
            \draw[very thick] (-3.5,2) -- (-2.6,2);
            \draw[very thick] (-3.5,-2) -- (-2.6,-2);
            \draw[very thick] (-5.5,-2.9) -- (-5.5,2.9);
            \draw[very thick] (-5.5,-4.5) -- (-5.5,-3.6);
            \draw[very thick] (-5.5,4.5) -- (-5.5,3.6);
            
            \draw[very thick] (-3.5,-2.9) -- (-3.5,-2);
            \draw[very thick] (-3.5,2) -- (-3.5,2.9);
            \draw[very thick] (-3.5,-4.5) -- (-3.5,-3.6);
            \draw[very thick] (-3.5,4.5) -- (-3.5,3.6);
            
            \draw[very thick] (-7.5,-2.9) -- (-7.5,2.9);
            \draw[very thick] (-7.5,-4.5) -- (-7.5,-3.6);
            \draw[very thick] (-7.5,4.5) -- (-7.5,3.6);
            
            \foreach \i in {-6,-5,-3,-2} {
                \foreach \j in {9,8,6,5,-4,-7,-5,-8} {
                    \darr{\j}{\i}{\ar};
                }
                \darr{4}{\i}{\ard};
                \darr{3}{\i}{\ad};
                \darr{2}{\i}{\ald};
                \darr{1}{\i}{\al};
                
                \darr{-3}{\i}{\aru};
                \darr{-2}{\i}{\au};
                \darr{-1}{\i}{\alu};
                \darr{0}{\i}{\al};
            }
            
            \darrs{9} {\al}{\ald}{\ad}{\ard}
            \darrs{8} {\al}{\ald}{\ad}{\ard}
            \darrs{6} {\al}{\ald}{\ad}{\ard}
            \darrs{5} {\al}{\ald}{\ad}{\ard}
            
            \darrs{4} {\al}{\ald}{\ad}{\ard}
            \darrs{3} {\al}{\ald}{\ad}{\ad}
            \darrs{2} {\al}{\ald}{\ald}{\ald}
            \darrs{1} {\al}{\al}{\al}{\al}
            \darrs{0} {\al}{\al}{\al}{\al}
            \darrs{-1}{\al}{\alu}{\alu}{\alu}
            \darrs{-2}{\al}{\alu}{\au}{\au}
            \darrs{-3}{\al}{\alu}{\au}{\aru}
            
            \darrs{-4} {\al}{\alu}{\au}{\aru}
            \darrs{-5} {\al}{\alu}{\au}{\aru}
            \darrs{-7} {\al}{\alu}{\au}{\aru}
            \darrs{-8} {\al}{\alu}{\au}{\aru}
            
            \foreach \i in {9,8,6,5,...,-5,-7,-8} {
                \darr{\i}{3}{\al};
            }
            
            \draw[snake=brace, mirror snake, segment amplitude=2mm] (-1, -4.6) -- ++(2, 0);
            \draw (0, -5) node {\bf  Old Buffer};
            
            \draw[snake=brace, mirror snake, segment amplitude=2mm] (-7.5, -4.6) -- ++(2, 0);
            \draw (-6.5, -5) node {\bf  New Buffer};
            
            \draw (0, 0) node[rotate=70] {\bf \large Pattern 1\slash2};
        \end{scope}

        \begin{scope}[xshift=4cm]
            \draw[step=.5cm,style=help lines] (-1,-4.5) grid (1,-3.6);
            \draw[step=.5cm,style=help lines] (-1,-2.9) grid (1, 2.9);
            \draw[very thick] (-1,-2) rectangle (1,2);
            \draw[step=.5cm,style=help lines] (-1,3.6) grid (1, 4.5);
            
            \draw[snake=brace, mirror snake, segment amplitude=2mm] (-1, -4.6) -- ++(2, 0);
            \draw (0, -5) node {\bf Pattern 1};
            
            \darrl{9}{\ar}{\ard}{\ad}{\ald};
            \darrl{8}{\ar}{\ard}{\ad}{\ald};
            % ...
            \darrl{7}{\cdots}{\cdots}{\cdots}{\cdots};
            
            \darrl{6}{\ar}{\ard}{\ad}{\ald};
            \darrl{5}{\ar}{\ard}{\ad}{\ald};;
            \darrl{4} {\ard}{\ad}{\ald}{\ald};
            \darrl{3} {\ad}{\ald}{\ald}{\ald};
            \darrl{2} {\ald}{\ald}{\ald}{\ald};
            \darrl{1} {\ald}{\ald}{\ald}{\ald};
            
            \darrl{0} {\al}{\al}{\al}{\al};
            \darrl{-1}{\alu}{\al}{\al}{\al};
            \darrl{-2}{\au}{\alu}{\al}{\al};
            \darrl{-3}{\aru}{\au}{\alu}{\al};
            \darrl{-4}{\ar}{\aru}{\au}{\alu};
            \darrl{-5}{\ar}{\aru}{\au}{\alu};
            % ...
            \darrl{-6}{\cdots}{\cdots}{\cdots}{\cdots};
            
            \darrl{-7}{\ar}{\aru}{\au}{\alu};
            \darrl{-8}{\ar}{\aru}{\au}{\alu};
        \end{scope}
        
        \begin{scope}[xshift=7cm]
            \draw[step=.5cm,style=help lines] (-1,-4.5) grid (1,-3.6);
            \draw[step=.5cm,style=help lines] (-1,-2.9) grid (1, 2.9);
            \draw[very thick] (-1,-2) rectangle (1,2);
            \draw[step=.5cm,style=help lines] (-1,3.6) grid (1, 4.5);
            
            \draw[snake=brace, mirror snake, segment amplitude=2mm] (-1, -4.6) -- ++(2, 0);
            \draw (0, -5) node {\bf Pattern 2};
            
            \darrl{9}{\ard}{\ad}{\ald}{\al};
            \darrl{8}{\ard}{\ad}{\ald}{\al};
            % ...
            \darrl{7}{\cdots}{\cdots}{\cdots}{\cdots};
            
            \darrl{6}{\ard}{\ad}{\ald}{\al};
            \darrl{5}{\ard}{\ad}{\ald}{\al};
            \darrl{4} {\ad}{\ald}{\al}{\al};
            \darrl{3} {\ald}{\al}{\al}{\al};
            \darrl{2} {\al}{\al}{\al}{\al};
            \darrl{1} {\al}{\al}{\al}{\al};
            
            \darrl{0} {\alu}{\alu}{\alu}{\alu};
            \darrl{-1}{\au}{\alu}{\alu}{\alu};
            \darrl{-2}{\aru}{\au}{\alu}{\alu};
            \darrl{-3}{\ar}{\aru}{\au}{\alu};
            \darrl{-4}{\aru}{\au}{\alu}{\al};
            \darrl{-5}{\aru}{\au}{\alu}{\al};
            % ...
            \darrl{-6}{\cdots}{\cdots}{\cdots}{\cdots};
            
            \darrl{-7}{\aru}{\au}{\alu}{\al};
            \darrl{-8}{\aru}{\au}{\alu}{\al};
        \end{scope}
        
    \end{tikzpicture}
\vspace{0.3cm}

%% file: parts/ppad.tex
\section{Arrows is $\PPAD$-complete}

In this section we revisit the problems from the previous one, but with another computational model.
In the first step we will give basic definitions of computational classes $\FNP$, $\TFNP$ and $\PPAD$.
It is already known that $\TFNP$-problem, based on Sperner's lemma, is $\PPAD$-complete~\cite{ChenD06}.
This section is devoted to $\TFNP$-problem, based on arrows' game. We denote it as $\ARROWS$ and give the strict definition below.
The main result, presented in this section, that $\ARROWS$ is $\PPAD$-complete.

\begin{definition}[$\FNP$ and $\TFNP$]
    Let $R \subset \Sigma^* \times \Sigma^*$ be a polynomial-time computable relation such that there exists a polynomial $p$ that for every $(x, y) \in R$ we have $y \leq p(|x|)$. The $\NP$ search problem $Q_R$ specified by $R$ is to given input $x \in \Sigma^*$ find $y \in \Sigma^*$ such that $(x, y) \in R$, if such $y$ exists, or return `no' otherwise. We use $\FNP$ to denote the class of $\NP$ search problems. An $\NP$ search problem is said to be total if for every $x$, there exists an $y$ such that $(x, y) \in R$. We use $\TFNP$ to denote the class of total $\NP$ search problems.
\end{definition}

\begin{definition}[Polynomial Reduction]
    A search problem $Q_{R_1} \in \TFNP$ is polynomial-time reducible to a search problem $Q_{R_2} \in \TFNP$ if there exists a pair of  polynomial-time computable functions $(f, g)$ such that for every input $x$ in $Q_{R_1}$, if $y$ satisfies $(f(x), y) \in R_2$, then $(x, g(y)) \in R_1$.
\end{definition}

\begin{definition}[$\LEAFD$]
    The input of the problem is a pair  $(M, 0^k)$ where $M$ is the description of a polynomial-time Turing machine which satisfies 
    \begin{enumerate}
        \item for every $v \in \{0, 1\}^k$, $M(v)$ is an ordered pair $(u_1, u_2)$ where $u_1, u_2 \in \{0, 1\}^k \cup \{ no \}$;
        \item $M(0^k) = (no, 1^k)$ and the first component of $M(1^k)$ is $0^k$.
    \end{enumerate}
    $M$ generates a directed graph $G = \langle V, E \rangle$ where $V = \{0, 1\}^k$ in the following way. An edge $uv$ appears in $E$ iff $v$ is the second component of $M(u)$ and $u$ is the first component of $M(v)$.
    The output is a directed leaf (in-degree + out-degree = 1) of graph $G$ which is different from $0^k$.
\end{definition}

$\PPAD$ \cite{Pap94} is the set of total $\NP$ search problems that are polynomial-time reducible to $\LEAFD$. By definition, $\LEAFD$ is complete for $\PPAD$.

For every $n \geq 1$, let $B_n = \{ \textbf{p} = (p_1, p_2) \in \mathbb{Z}^2 | 0 \leq p_1 < n \text{ and } 0 \leq p_2 < n \}$.

The boundary of $B_n$ is then the set of points $\textbf{p} \in B_n$ with $p_i \in \{0, n - 1\}$ for some $i \in \{1, 2\}$. For every $p \in \mathbb{Z}^2$ , we define $K_p = \{ q = (q_1 , q_2) \in \mathbb{Z}^2 | q_1 \in [p_1, p_1 + 1] \text{ and } q_2 \in [p_2, p_2 + 1]\}$, the $2 \times 2$ square with left bottom corner at $p$.

An arrow assignment of $B_n$ is a function $f$ from $B_n$ to $A = \{\au, \alu, \al, \ald, \ad, \ard, \ar, \aru\}$. It is said to be valid if for every $p \in B_n$ on the boundary $B_n$
\begin{enumerate}
    \item if $p_2 = 0$, then $f(p) \in \{\aru, \au, \alu\}$;
    \item if $p_2 = n - 1$, then $f(p) \in \{\ard, \ad, \ald\}$;
    \item if $p_1 = 0$, then $f(p) \in \{\ard, \ar, \aru\}$;
    \item if $p_1 = n - 1$, then $f(p) \in \{\ald, \al, \alu\}$.
\end{enumerate}
Notice that for every corner there is only one valid arrow.

\begin{definition}[$\ARROWS$]
The input of the problem $\ARROWS$ is a pair $(F, 0^k)$ where $F$ is the description of a polynomial-time Turing machine which generates a valid arrow assignment $f$ on $B_{2^k}$. Here $f(p) = F(p) \in A$ for every $p \in B_{2^k}$. The output is a pair of points $\textbf{(p, q)} \in B_{2^k}^2$ such that $p$ and $q$ are adjacent by edge and the angle between their arrows more than 45 degree.  
\end{definition}

In order to prove that $\ARROWS \in \PPAD$ we use $\BROUWERD$ problem which is very similar to $\ARROWS$. It is known that $\BROUWERD$ is $\PPAD$-complete \cite{ChenD06}.

A 3-coloring of $B_n$ is a function $g$ from $B_n$ to $\{0, 1, 2\}$. It is said to be valid if for every $\textbf{p}$ on the boundary of $B_n$,
\begin{enumerate}
    \item if $p_2 = 0$, then $g(p) = 2$;
    \item if $p_2 \neq 0$ and $p_1 = 0$, then $g(p) = 0$;
    \item otherwise, $g(p) = 1$.
\end{enumerate}
The search problem $\BROUWERD$ is then defined as follows.

\begin{definition}[$\BROUWERD$, \cite{ChenD06}]
The input of the problem $\BROUWERD$ is a pair $(F, 0^k)$ where $F$ is the description of a polynomial-time Turing machine which generates a valid 3-coloring $g$ on $B_{2^k}$. Here $g(p) = F(p) \in \{0, 1, 2\}$ for every $p \in B_{2^k}$. The output is a point $\textbf{p} \in B_{2^k}$ such that $K_p$ is trichromatic, that is, $K_p$ has all the three colors.
\end{definition}

\begin{lemma}
\label{arrows-ppad}
$\ARROWS \in \PPAD$.
\end{lemma}
\begin{proof}
    We divide all arrows in three groups $\{\ad, \ald, \al\}$, $\{\ar, \ard\}$, $\{\aru, \au, \alu\}$ and assign each group value $0$, $1$, and $2$, respectively.

    \begin{figure}[h]
        \center{\scalebox{0.6}{\input{pics/a-bound.tex}}}
        \caption{Square boundary \label{pic:bound-gadget}}
    \end{figure}
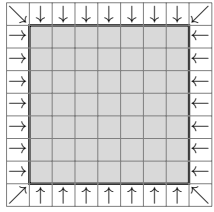

    Now we surround the current square by arrows as it is shown on Figure~\ref{pic:bound-gadget}, and convert new instance to $\BROUWERD$ problem replacing arrows by values $0$, $1$, and $2$ as they are divided on groups above.

    Notice that built coloring is valid, and there should be a trichromatic square $K_p$, which will be an output of $\BROUWERD$ problem. By construction square $K_p$ can not touch the border of the big square.

    On the other hand if square $2 \times 2$ is trichromatic then there should be two cells in it whose arrows directions differ 
on more than 45 degree. The latter can be easily checked since 
there is no arrow in one of the groups that is close up to 45 degrees to some arrows in two other groups.
\end{proof}

To prove $\PPAD$-hardness we use planar version of $\LEAFD$ problem.

For every $n \geq 1$ let us denote $V_n = \{ u = (u_1, u_2) \,|\, 0 \leq u_1 < n \text{ and } 0 \leq u_2 < n \}$.

\begin{definition}[$\RLEAFD$]
The input instance is a pair $(K, 0^k)$ where $K$ is the description of a polynomial-time Turing machine which satisfies:
\begin{enumerate}
    \item For every $u \in V_{2^k}$  $K(u)$ is an ordered pair $(u_1, u_2)$ where $u_1, u_2 \in V_{2^k} \cup \{no\}$;
    \item $K((0, 0))$ = $(no,(1, 0))$ and the first component of $K((1, 0))$ is $(0, 0)$.
\end{enumerate}
$K$ generates a directed graph $G = (V_{2^k}, E)$ on the grid in the following way. 
$(u,v)$ is an edge iff $v$ is the second component of $K(u)$, $u$ is the first component of $K(v)$ 
and $|u_1 - v_1| + |u_2 - v_2| = 1$.

The output is a directed leaf (with in-degree + out-degree = 1) of graph $G$ which is different from the origin $(0, 0)$.
\end{definition}

Chen and Deng proved that $\RLEAFD$ problem is $\PPAD$-complete \cite{ChenD06}.

\begin{lemma}
\label{arrows-ppad-hard}
$\ARROWS$ is $\PPAD$-hard.
\end{lemma}
\begin{proof}

We build a polynomial reduction from $\RLEAFD$ to $\ARROWS$ using the following gadgets. 
Let $G = \langle V_{2^k}, E \rangle$ be a graph from the defition of $\RLEAFD$. We build the following grid as it is shown on 
Figure~\ref{pic:border-gadget}. Empty rectangle on the figure contains $2^k \times 2^k$ blocks, where every block is a $9 \times 9$ square.
 We fill block $(i, j)$ according the information about vertex $(i, j)$ in graph $G$ and edges that are incident to $(i,j)$ 
(see Figure~\ref{pic:block-gadget}). 

    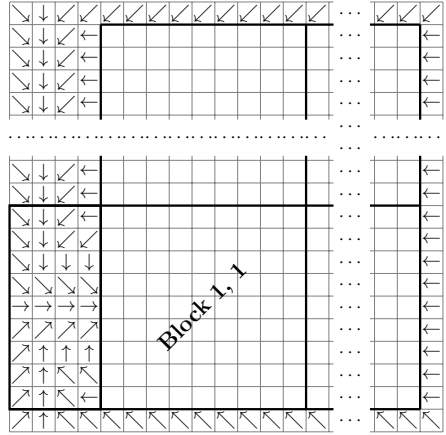
\begin{figure}[h]
		\center{\scalebox{0.6}{\input{pics/a-place.tex}}}
		\caption{Border gadget \label{pic:border-gadget}}
    \end{figure}

    \begin{figure}[h]
		\center{\scalebox{0.6}{\input{pics/a-blocks.tex}}}
		\caption{Block gadget \label{pic:block-gadget}}
    \end{figure}

If the degree of the node is zero, then each cell in the block contains $\al$. 
Otherwise we build a path of $\ar$ that goes through the center of the block in direction of the corresponding path in graph $G$ (see
Figure~\ref{pic:example}).
 
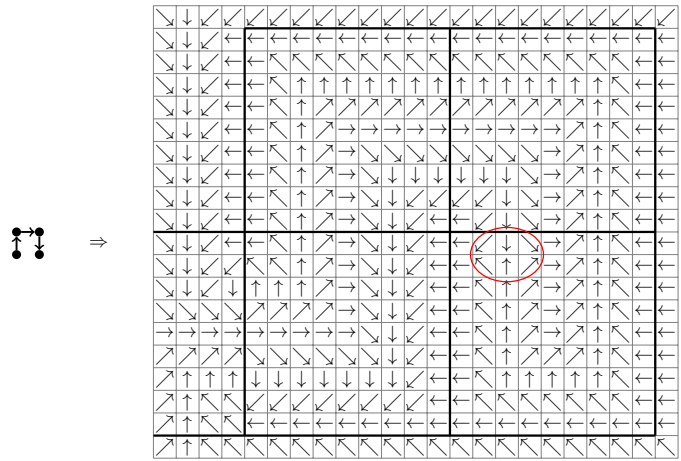
    \begin{figure}[h]
        \center{\scalebox{0.6}{\input{pics/a-example.tex}}}
        \caption{Example \label{pic:example}}
    \end{figure}

It is easy to see that there are no any conflicts on the borders of the blocks and with the border gadget. 
So the only places with conflicts are ends of the paths in graph $G$. 
A value in any cell can be computed in polynomial time using a description of a polynomial-time Turing Machine
 from the instance of $\RLEAFD$; it is also easy to compute in polynomial time a leaf in $\RLEAFD$ from the coordinates 
of the conflict in $\ARROWS$. 

\end{proof}

The proof of the next theorem follows from Lemma~\ref{arrows-ppad} and \ref{arrows-ppad-hard}.

\begin{theorem}
$\ARROWS$ is $\PPAD$-complete.
\end{theorem}

%% file: pics/a-bound.tex
\vspace{0.3cm}
{
    \tikzstyle{arr}=[xshift=7, yshift=7]
    \newcommand{\darr}[3]{\node[arr] at (#2 / 2, #1 / 2) {$#3$}}
    \newcommand{\darrl}[9]{\darr{#1}{0}{#2};\darr{#1}{1}{#3};\darr{#1}{2}{#4};\darr{#1}{3}{#5};\darr{#1}{4}{#6};\darr{#1}{5}{#7};\darr{#1}{6}{#8};\darr{#1}{7}{#9};}
    
    \newcommand{\xstep}{8cm}
    \newcommand{\ystep}{6cm}
    
    \begin{tikzpicture}[scale=1]
        \begin{scope}
        
            \draw[line width=0.05cm, fill=gray!30!white] (0.5, 0.5) rectangle (4.0, 4.0);
            \draw[step=.5cm,style=help lines] (0, 0) grid (4.5, 4.5);

            \darr{0}{0}{\aru};
            \darr{8}{0}{\ard};
            \darr{0}{8}{\alu};
            \darr{8}{8}{\ald};
            \foreach \x in {1, ..., 7} {
                \darr{\x}{0}{\ar};
                \darr{\x}{8}{\al};
                \darr{0}{\x}{\au};
                \darr{8}{\x}{\ad};
            }

        \end{scope}

    \end{tikzpicture}
}
\vspace{0.3cm}

%% file: pics/a-place.tex
\vspace{0.3cm}
{
    \tikzstyle{lw}=[very thick, line width=0.5mm]
    \tikzstyle{arr}=[xshift=7, yshift=7]
    \newcommand{\darr}[3]{\node[arr] at (#1 / 2, #2 / 2) {$#3$}}
    \newcommand{\darrl}[4]{\darr{0}{#1}{#2};\darr{1}{#1}{#3};\darr{2}{#1}{#4};}
    \newcommand{\darrs}[5]{\darr{0}{#1}{#2};\darr{1}{#1}{#3};\darr{2}{#1}{#4};\darr{3}{#1}{#5};}

    \begin{tikzpicture}[scale=1]
        \def\W{18}
        \def\H{18}
        \begin{scope}
            \draw[step=.5cm,style=help lines] (0, 0) grid (13 * 0.5 + 0.6, 11 * 0.5 + 0.6);
            \draw[step=.5cm,style=help lines] (14 * 0.5 + 0.9, 12 * 0.5 + 0.9) grid (\W * 0.5 + 0.5, \H * 0.5 + 0.5);
            \draw[step=.5cm,style=help lines] (0, 12 * 0.5 + 0.9) grid (13 * 0.5 + 0.6, \H * 0.5 + 0.5);
            \draw[step=.5cm,style=help lines] (14 * 0.5 + 0.9, 0) grid (\W * 0.5 + 0.5, 11 * 0.5 + 0.6);

            \draw[line width=0.5mm] (0, 0.5) rectangle (2, 5);
            
            \darr{0}{\H}{\ard}; \darr{1}{\H}{\ad};\darr{2}{\H}{\ald};
            \foreach \i in {3,...,13,16,17,18} {
                \darr{\i}{\H}{\ald};
            }

            \foreach \i in {9,10,11,17,14,15,16} {
                \darr{0}{\i}{\ard}; \darr{1}{\i}{\ad};\darr{2}{\i}{\ald};\darr{3}{\i}{\al};
            }
            \foreach \i in {1,...,11,17, 14, 15, 16} {
                \darr{\W}{\i}{\al};
            }
            \foreach \i in {3,...,13,17,16} {
                \darr{\i}{0}{\alu};
            }
            
            \foreach \i in {0,...,13, 18, 16, 17} {
                \darr{\i}{12.5}{\cdots};
            }

            \foreach \i in {0,...,18} {
                \darr{14.5}{\i}{\cdots};
            }

            \darrl{0}{\aru}{\au}{\alu};
            \darrs{1}{\aru}{\au}{\alu}{\al};            
            
            \darrs{8}{\ard}{\ad}{\ald}{\ald};
            \darrs{7}{\ard}{\ad}{\ad}{\ad};
            \darrs{6}{\ard}{\ard}{\ard}{\ard};
            \darrs{5}{\ar}{\ar}{\ar}{\ar};
            \darrs{4}{\aru}{\aru}{\aru}{\aru};
            \darrs{3}{\aru}{\au}{\au}{\au};
            \darrs{2}{\aru}{\au}{\alu}{\alu};

            \darr{\W}{0}{\alu};            
            
            \draw[lw] (0, 0.5) -- (7.1, 0.5);
            \draw[lw] (7.9, 0.5) -- (9,0.5);
            \draw[lw] (0, 5) -- (7.1, 5);
            \draw[lw] (7.9, 5) -- (9, 5);
            \draw[lw] (2, 9) -- (7.1, 9);
            \draw[lw] (7.9, 9) -- (9, 9);

            \draw[lw] (2, 0.5) -- (2, 6.1); \draw[lw] (2, 6.9) -- (2, 9);
            \draw[lw] (6.5, 0.5) -- (6.5, 6.1); \draw[lw] (6.5, 6.9) -- (6.5, 9);
            \draw[lw] (9, 0.5) -- (9, 6.1); \draw[lw] (9, 6.9) -- (9, 9);

            \draw (4.25, 2.75) node[rotate=45] {\bf \large Block 1, 1};

        \end{scope}
        
    \end{tikzpicture}
}
\vspace{0.3cm}

%% file: pics/a-blocks.tex
\vspace{0.3cm}
{
    \tikzstyle{arr}=[xshift=7, yshift=7]
    \newcommand{\darr}[3]{\node[arr] at (#2, #1) {$#3$}}
    \newcommand{\darrl}[9]{\darr{#1}{0}{#2};\darr{#1}{1}{#3};\darr{#1}{2}{#4};\darr{#1}{3}{#5};\darr{#1}{4}{#6};\darr{#1}{5}{#7};\darr{#1}{6}{#8};\darr{#1}{7}{#9};}
    
    \newcommand{\xstep}{8cm}
    \newcommand{\ystep}{6cm}

    \begin{tikzpicture}[scale=1]

        \begin{scope}[xshift= 0 * \xstep, yshift = 6 * \ystep]       
            \begin{scope}[yshift=2mm, xshift=-3mm]
                \draw[very thick, ->] (-1.0, 2.0) -> (-1.4, 2.0);
                \draw[very thick, ->] (-1.5, 1.9) -> (-1.5, 1.5);
                \node[fill=black, circle, inner sep=0.7mm] at (-1.5, 2) {};

                \node[xshift=3mm] at (-0.7, 2) {$\Rightarrow$};
            \end{scope}

            \draw[step=.5cm,style=help lines] (0, 0) grid (4.5, 4.5);

            \darr{0.0}{0.0}{\al};\darr{0.0}{0.5}{\ald};\darr{0.0}{1.0}{\ad};\darr{0.0}{1.5}{\ard};\darr{0.0}{2.0}{\ar};\darr{0.0}{2.5}{\aru};\darr{0.0}{3.0}{\au};\darr{0.0}{3.5}{\alu};\darr{0.0}{4.0}{\al};
\darr{0.5}{0.0}{\al};\darr{0.5}{0.5}{\ald};\darr{0.5}{1.0}{\ad};\darr{0.5}{1.5}{\ard};\darr{0.5}{2.0}{\ar};\darr{0.5}{2.5}{\aru};\darr{0.5}{3.0}{\au};\darr{0.5}{3.5}{\alu};\darr{0.5}{4.0}{\alu};
\darr{1.0}{0.0}{\al};\darr{1.0}{0.5}{\ald};\darr{1.0}{1.0}{\ad};\darr{1.0}{1.5}{\ard};\darr{1.0}{2.0}{\ar};\darr{1.0}{2.5}{\aru};\darr{1.0}{3.0}{\au};\darr{1.0}{3.5}{\au};\darr{1.0}{4.0}{\au};
\darr{1.5}{0.0}{\al};\darr{1.5}{0.5}{\ald};\darr{1.5}{1.0}{\ad};\darr{1.5}{1.5}{\ard};\darr{1.5}{2.0}{\ar};\darr{1.5}{2.5}{\aru};\darr{1.5}{3.0}{\aru};\darr{1.5}{3.5}{\aru};\darr{1.5}{4.0}{\aru};
\darr{2.0}{0.0}{\al};\darr{2.0}{0.5}{\ald};\darr{2.0}{1.0}{\ad};\darr{2.0}{1.5}{\ard};\darr{2.0}{2.0}{\ar};\darr{2.0}{2.5}{\ar};\darr{2.0}{3.0}{\ar};\darr{2.0}{3.5}{\ar};\darr{2.0}{4.0}{\ar};
\darr{2.5}{0.0}{\al};\darr{2.5}{0.5}{\ald};\darr{2.5}{1.0}{\ad};\darr{2.5}{1.5}{\ard};\darr{2.5}{2.0}{\ard};\darr{2.5}{2.5}{\ard};\darr{2.5}{3.0}{\ard};\darr{2.5}{3.5}{\ard};\darr{2.5}{4.0}{\ard};
\darr{3.0}{0.0}{\al};\darr{3.0}{0.5}{\ald};\darr{3.0}{1.0}{\ad};\darr{3.0}{1.5}{\ad};\darr{3.0}{2.0}{\ad};\darr{3.0}{2.5}{\ad};\darr{3.0}{3.0}{\ad};\darr{3.0}{3.5}{\ad};\darr{3.0}{4.0}{\ad};
\darr{3.5}{0.0}{\al};\darr{3.5}{0.5}{\ald};\darr{3.5}{1.0}{\ald};\darr{3.5}{1.5}{\ald};\darr{3.5}{2.0}{\ald};\darr{3.5}{2.5}{\ald};\darr{3.5}{3.0}{\ald};\darr{3.5}{3.5}{\ald};\darr{3.5}{4.0}{\ald};
\darr{4.0}{0.0}{\al};\darr{4.0}{0.5}{\al};\darr{4.0}{1.0}{\al};\darr{4.0}{1.5}{\al};\darr{4.0}{2.0}{\al};\darr{4.0}{2.5}{\al};\darr{4.0}{3.0}{\al};\darr{4.0}{3.5}{\al};\darr{4.0}{4.0}{\al};

        \end{scope}

        \begin{scope}[xshift= 1 * \xstep, yshift = 6 * \ystep]       
            \begin{scope}[yshift=2mm, xshift=-3mm]
                \draw[very thick, ->] (-1.0, 2.0) -> (-1.4, 2.0);
                \draw[very thick, ->] (-1.6, 2.0) -> (-2.0, 2.0);
                \node[fill=black, circle, inner sep=0.7mm] at (-1.5, 2) {};

                \node[xshift=3mm] at (-0.7, 2) {$\Rightarrow$};
            \end{scope}

            \draw[step=.5cm,style=help lines] (0, 0) grid (4.5, 4.5);

            \darr{0.0}{0.0}{\al};\darr{0.0}{0.5}{\al};\darr{0.0}{1.0}{\al};\darr{0.0}{1.5}{\al};\darr{0.0}{2.0}{\al};\darr{0.0}{2.5}{\al};\darr{0.0}{3.0}{\al};\darr{0.0}{3.5}{\al};\darr{0.0}{4.0}{\al};
\darr{0.5}{0.0}{\alu};\darr{0.5}{0.5}{\alu};\darr{0.5}{1.0}{\alu};\darr{0.5}{1.5}{\alu};\darr{0.5}{2.0}{\alu};\darr{0.5}{2.5}{\alu};\darr{0.5}{3.0}{\alu};\darr{0.5}{3.5}{\alu};\darr{0.5}{4.0}{\alu};
\darr{1.0}{0.0}{\au};\darr{1.0}{0.5}{\au};\darr{1.0}{1.0}{\au};\darr{1.0}{1.5}{\au};\darr{1.0}{2.0}{\au};\darr{1.0}{2.5}{\au};\darr{1.0}{3.0}{\au};\darr{1.0}{3.5}{\au};\darr{1.0}{4.0}{\au};
\darr{1.5}{0.0}{\aru};\darr{1.5}{0.5}{\aru};\darr{1.5}{1.0}{\aru};\darr{1.5}{1.5}{\aru};\darr{1.5}{2.0}{\aru};\darr{1.5}{2.5}{\aru};\darr{1.5}{3.0}{\aru};\darr{1.5}{3.5}{\aru};\darr{1.5}{4.0}{\aru};
\darr{2.0}{0.0}{\ar};\darr{2.0}{0.5}{\ar};\darr{2.0}{1.0}{\ar};\darr{2.0}{1.5}{\ar};\darr{2.0}{2.0}{\ar};\darr{2.0}{2.5}{\ar};\darr{2.0}{3.0}{\ar};\darr{2.0}{3.5}{\ar};\darr{2.0}{4.0}{\ar};
\darr{2.5}{0.0}{\ard};\darr{2.5}{0.5}{\ard};\darr{2.5}{1.0}{\ard};\darr{2.5}{1.5}{\ard};\darr{2.5}{2.0}{\ard};\darr{2.5}{2.5}{\ard};\darr{2.5}{3.0}{\ard};\darr{2.5}{3.5}{\ard};\darr{2.5}{4.0}{\ard};
\darr{3.0}{0.0}{\ad};\darr{3.0}{0.5}{\ad};\darr{3.0}{1.0}{\ad};\darr{3.0}{1.5}{\ad};\darr{3.0}{2.0}{\ad};\darr{3.0}{2.5}{\ad};\darr{3.0}{3.0}{\ad};\darr{3.0}{3.5}{\ad};\darr{3.0}{4.0}{\ad};
\darr{3.5}{0.0}{\ald};\darr{3.5}{0.5}{\ald};\darr{3.5}{1.0}{\ald};\darr{3.5}{1.5}{\ald};\darr{3.5}{2.0}{\ald};\darr{3.5}{2.5}{\ald};\darr{3.5}{3.0}{\ald};\darr{3.5}{3.5}{\ald};\darr{3.5}{4.0}{\ald};
\darr{4.0}{0.0}{\al};\darr{4.0}{0.5}{\al};\darr{4.0}{1.0}{\al};\darr{4.0}{1.5}{\al};\darr{4.0}{2.0}{\al};\darr{4.0}{2.5}{\al};\darr{4.0}{3.0}{\al};\darr{4.0}{3.5}{\al};\darr{4.0}{4.0}{\al};

        \end{scope}

        \begin{scope}[xshift= 2 * \xstep, yshift = 6 * \ystep]       
            \begin{scope}[yshift=2mm, xshift=-3mm]
                \draw[very thick, ->] (-1.5, 1.5) -> (-1.5, 1.9);
                \draw[very thick, ->] (-1.6, 2.0) -> (-2.0, 2.0);
                \node[fill=black, circle, inner sep=0.7mm] at (-1.5, 2) {};

                \node[xshift=3mm] at (-0.7, 2) {$\Rightarrow$};
            \end{scope}

            \draw[step=.5cm,style=help lines] (0, 0) grid (4.5, 4.5);

            \darr{0.0}{0.0}{\al};\darr{0.0}{0.5}{\alu};\darr{0.0}{1.0}{\au};\darr{0.0}{1.5}{\aru};\darr{0.0}{2.0}{\ar};\darr{0.0}{2.5}{\ard};\darr{0.0}{3.0}{\ad};\darr{0.0}{3.5}{\ald};\darr{0.0}{4.0}{\al};
\darr{0.5}{0.0}{\alu};\darr{0.5}{0.5}{\alu};\darr{0.5}{1.0}{\au};\darr{0.5}{1.5}{\aru};\darr{0.5}{2.0}{\ar};\darr{0.5}{2.5}{\ard};\darr{0.5}{3.0}{\ad};\darr{0.5}{3.5}{\ald};\darr{0.5}{4.0}{\al};
\darr{1.0}{0.0}{\au};\darr{1.0}{0.5}{\au};\darr{1.0}{1.0}{\au};\darr{1.0}{1.5}{\aru};\darr{1.0}{2.0}{\ar};\darr{1.0}{2.5}{\ard};\darr{1.0}{3.0}{\ad};\darr{1.0}{3.5}{\ald};\darr{1.0}{4.0}{\al};
\darr{1.5}{0.0}{\aru};\darr{1.5}{0.5}{\aru};\darr{1.5}{1.0}{\aru};\darr{1.5}{1.5}{\aru};\darr{1.5}{2.0}{\ar};\darr{1.5}{2.5}{\ard};\darr{1.5}{3.0}{\ad};\darr{1.5}{3.5}{\ald};\darr{1.5}{4.0}{\al};
\darr{2.0}{0.0}{\ar};\darr{2.0}{0.5}{\ar};\darr{2.0}{1.0}{\ar};\darr{2.0}{1.5}{\ar};\darr{2.0}{2.0}{\ar};\darr{2.0}{2.5}{\ard};\darr{2.0}{3.0}{\ad};\darr{2.0}{3.5}{\ald};\darr{2.0}{4.0}{\al};
\darr{2.5}{0.0}{\ard};\darr{2.5}{0.5}{\ard};\darr{2.5}{1.0}{\ard};\darr{2.5}{1.5}{\ard};\darr{2.5}{2.0}{\ard};\darr{2.5}{2.5}{\ard};\darr{2.5}{3.0}{\ad};\darr{2.5}{3.5}{\ald};\darr{2.5}{4.0}{\al};
\darr{3.0}{0.0}{\ad};\darr{3.0}{0.5}{\ad};\darr{3.0}{1.0}{\ad};\darr{3.0}{1.5}{\ad};\darr{3.0}{2.0}{\ad};\darr{3.0}{2.5}{\ad};\darr{3.0}{3.0}{\ad};\darr{3.0}{3.5}{\ald};\darr{3.0}{4.0}{\al};
\darr{3.5}{0.0}{\ald};\darr{3.5}{0.5}{\ald};\darr{3.5}{1.0}{\ald};\darr{3.5}{1.5}{\ald};\darr{3.5}{2.0}{\ald};\darr{3.5}{2.5}{\ald};\darr{3.5}{3.0}{\ald};\darr{3.5}{3.5}{\ald};\darr{3.5}{4.0}{\al};
\darr{4.0}{0.0}{\al};\darr{4.0}{0.5}{\al};\darr{4.0}{1.0}{\al};\darr{4.0}{1.5}{\al};\darr{4.0}{2.0}{\al};\darr{4.0}{2.5}{\al};\darr{4.0}{3.0}{\al};\darr{4.0}{3.5}{\al};\darr{4.0}{4.0}{\al};

        \end{scope}

        \begin{scope}[xshift= 0 * \xstep, yshift = 5 * \ystep]       
            \begin{scope}[yshift=2mm, xshift=-3mm]
                \draw[very thick, ->] (-1.5, 2.5) -> (-1.5, 2.1);
                \draw[very thick, ->] (-1.5, 1.9) -> (-1.5, 1.5);
                \node[fill=black, circle, inner sep=0.7mm] at (-1.5, 2) {};

                \node[xshift=3mm] at (-0.7, 2) {$\Rightarrow$};
            \end{scope}

            \draw[step=.5cm,style=help lines] (0, 0) grid (4.5, 4.5);

            \darr{0.0}{0.0}{\al};\darr{0.0}{0.5}{\ald};\darr{0.0}{1.0}{\ad};\darr{0.0}{1.5}{\ard};\darr{0.0}{2.0}{\ar};\darr{0.0}{2.5}{\aru};\darr{0.0}{3.0}{\au};\darr{0.0}{3.5}{\alu};\darr{0.0}{4.0}{\al};
\darr{0.5}{0.0}{\al};\darr{0.5}{0.5}{\ald};\darr{0.5}{1.0}{\ad};\darr{0.5}{1.5}{\ard};\darr{0.5}{2.0}{\ar};\darr{0.5}{2.5}{\aru};\darr{0.5}{3.0}{\au};\darr{0.5}{3.5}{\alu};\darr{0.5}{4.0}{\al};
\darr{1.0}{0.0}{\al};\darr{1.0}{0.5}{\ald};\darr{1.0}{1.0}{\ad};\darr{1.0}{1.5}{\ard};\darr{1.0}{2.0}{\ar};\darr{1.0}{2.5}{\aru};\darr{1.0}{3.0}{\au};\darr{1.0}{3.5}{\alu};\darr{1.0}{4.0}{\al};
\darr{1.5}{0.0}{\al};\darr{1.5}{0.5}{\ald};\darr{1.5}{1.0}{\ad};\darr{1.5}{1.5}{\ard};\darr{1.5}{2.0}{\ar};\darr{1.5}{2.5}{\aru};\darr{1.5}{3.0}{\au};\darr{1.5}{3.5}{\alu};\darr{1.5}{4.0}{\al};
\darr{2.0}{0.0}{\al};\darr{2.0}{0.5}{\ald};\darr{2.0}{1.0}{\ad};\darr{2.0}{1.5}{\ard};\darr{2.0}{2.0}{\ar};\darr{2.0}{2.5}{\aru};\darr{2.0}{3.0}{\au};\darr{2.0}{3.5}{\alu};\darr{2.0}{4.0}{\al};
\darr{2.5}{0.0}{\al};\darr{2.5}{0.5}{\ald};\darr{2.5}{1.0}{\ad};\darr{2.5}{1.5}{\ard};\darr{2.5}{2.0}{\ar};\darr{2.5}{2.5}{\aru};\darr{2.5}{3.0}{\au};\darr{2.5}{3.5}{\alu};\darr{2.5}{4.0}{\al};
\darr{3.0}{0.0}{\al};\darr{3.0}{0.5}{\ald};\darr{3.0}{1.0}{\ad};\darr{3.0}{1.5}{\ard};\darr{3.0}{2.0}{\ar};\darr{3.0}{2.5}{\aru};\darr{3.0}{3.0}{\au};\darr{3.0}{3.5}{\alu};\darr{3.0}{4.0}{\al};
\darr{3.5}{0.0}{\al};\darr{3.5}{0.5}{\ald};\darr{3.5}{1.0}{\ad};\darr{3.5}{1.5}{\ard};\darr{3.5}{2.0}{\ar};\darr{3.5}{2.5}{\aru};\darr{3.5}{3.0}{\au};\darr{3.5}{3.5}{\alu};\darr{3.5}{4.0}{\al};
\darr{4.0}{0.0}{\al};\darr{4.0}{0.5}{\ald};\darr{4.0}{1.0}{\ad};\darr{4.0}{1.5}{\ard};\darr{4.0}{2.0}{\ar};\darr{4.0}{2.5}{\aru};\darr{4.0}{3.0}{\au};\darr{4.0}{3.5}{\alu};\darr{4.0}{4.0}{\al};

        \end{scope}

        \begin{scope}[xshift= 1 * \xstep, yshift = 5 * \ystep]       
            \begin{scope}[yshift=2mm, xshift=-3mm]

                \node[fill=black, circle, inner sep=0.7mm] at (-1.5, 2) {};

                \node[xshift=3mm] at (-0.7, 2) {$\Rightarrow$};
            \end{scope}

            \draw[step=.5cm,style=help lines] (0, 0) grid (4.5, 4.5);

            \darr{0.0}{0.0}{\al};\darr{0.0}{0.5}{\al};\darr{0.0}{1.0}{\al};\darr{0.0}{1.5}{\al};\darr{0.0}{2.0}{\al};\darr{0.0}{2.5}{\al};\darr{0.0}{3.0}{\al};\darr{0.0}{3.5}{\al};\darr{0.0}{4.0}{\al};
\darr{0.5}{0.0}{\al};\darr{0.5}{0.5}{\al};\darr{0.5}{1.0}{\al};\darr{0.5}{1.5}{\al};\darr{0.5}{2.0}{\al};\darr{0.5}{2.5}{\al};\darr{0.5}{3.0}{\al};\darr{0.5}{3.5}{\al};\darr{0.5}{4.0}{\al};
\darr{1.0}{0.0}{\al};\darr{1.0}{0.5}{\al};\darr{1.0}{1.0}{\al};\darr{1.0}{1.5}{\al};\darr{1.0}{2.0}{\al};\darr{1.0}{2.5}{\al};\darr{1.0}{3.0}{\al};\darr{1.0}{3.5}{\al};\darr{1.0}{4.0}{\al};
\darr{1.5}{0.0}{\al};\darr{1.5}{0.5}{\al};\darr{1.5}{1.0}{\al};\darr{1.5}{1.5}{\al};\darr{1.5}{2.0}{\al};\darr{1.5}{2.5}{\al};\darr{1.5}{3.0}{\al};\darr{1.5}{3.5}{\al};\darr{1.5}{4.0}{\al};
\darr{2.0}{0.0}{\al};\darr{2.0}{0.5}{\al};\darr{2.0}{1.0}{\al};\darr{2.0}{1.5}{\al};\darr{2.0}{2.0}{\al};\darr{2.0}{2.5}{\al};\darr{2.0}{3.0}{\al};\darr{2.0}{3.5}{\al};\darr{2.0}{4.0}{\al};
\darr{2.5}{0.0}{\al};\darr{2.5}{0.5}{\al};\darr{2.5}{1.0}{\al};\darr{2.5}{1.5}{\al};\darr{2.5}{2.0}{\al};\darr{2.5}{2.5}{\al};\darr{2.5}{3.0}{\al};\darr{2.5}{3.5}{\al};\darr{2.5}{4.0}{\al};
\darr{3.0}{0.0}{\al};\darr{3.0}{0.5}{\al};\darr{3.0}{1.0}{\al};\darr{3.0}{1.5}{\al};\darr{3.0}{2.0}{\al};\darr{3.0}{2.5}{\al};\darr{3.0}{3.0}{\al};\darr{3.0}{3.5}{\al};\darr{3.0}{4.0}{\al};
\darr{3.5}{0.0}{\al};\darr{3.5}{0.5}{\al};\darr{3.5}{1.0}{\al};\darr{3.5}{1.5}{\al};\darr{3.5}{2.0}{\al};\darr{3.5}{2.5}{\al};\darr{3.5}{3.0}{\al};\darr{3.5}{3.5}{\al};\darr{3.5}{4.0}{\al};
\darr{4.0}{0.0}{\al};\darr{4.0}{0.5}{\al};\darr{4.0}{1.0}{\al};\darr{4.0}{1.5}{\al};\darr{4.0}{2.0}{\al};\darr{4.0}{2.5}{\al};\darr{4.0}{3.0}{\al};\darr{4.0}{3.5}{\al};\darr{4.0}{4.0}{\al};

        \end{scope}

        \begin{scope}[xshift= 2 * \xstep, yshift = 5 * \ystep]       
            \begin{scope}[yshift=2mm, xshift=-3mm]
                \draw[very thick, ->] (-1.5, 1.5) -> (-1.5, 1.9);
                \draw[very thick, ->] (-1.5, 2.1) -> (-1.5, 2.5);
                \node[fill=black, circle, inner sep=0.7mm] at (-1.5, 2) {};

                \node[xshift=3mm] at (-0.7, 2) {$\Rightarrow$};
            \end{scope}

            \draw[step=.5cm,style=help lines] (0, 0) grid (4.5, 4.5);

            \darr{0.0}{0.0}{\al};\darr{0.0}{0.5}{\alu};\darr{0.0}{1.0}{\au};\darr{0.0}{1.5}{\aru};\darr{0.0}{2.0}{\ar};\darr{0.0}{2.5}{\ard};\darr{0.0}{3.0}{\ad};\darr{0.0}{3.5}{\ald};\darr{0.0}{4.0}{\al};
\darr{0.5}{0.0}{\al};\darr{0.5}{0.5}{\alu};\darr{0.5}{1.0}{\au};\darr{0.5}{1.5}{\aru};\darr{0.5}{2.0}{\ar};\darr{0.5}{2.5}{\ard};\darr{0.5}{3.0}{\ad};\darr{0.5}{3.5}{\ald};\darr{0.5}{4.0}{\al};
\darr{1.0}{0.0}{\al};\darr{1.0}{0.5}{\alu};\darr{1.0}{1.0}{\au};\darr{1.0}{1.5}{\aru};\darr{1.0}{2.0}{\ar};\darr{1.0}{2.5}{\ard};\darr{1.0}{3.0}{\ad};\darr{1.0}{3.5}{\ald};\darr{1.0}{4.0}{\al};
\darr{1.5}{0.0}{\al};\darr{1.5}{0.5}{\alu};\darr{1.5}{1.0}{\au};\darr{1.5}{1.5}{\aru};\darr{1.5}{2.0}{\ar};\darr{1.5}{2.5}{\ard};\darr{1.5}{3.0}{\ad};\darr{1.5}{3.5}{\ald};\darr{1.5}{4.0}{\al};
\darr{2.0}{0.0}{\al};\darr{2.0}{0.5}{\alu};\darr{2.0}{1.0}{\au};\darr{2.0}{1.5}{\aru};\darr{2.0}{2.0}{\ar};\darr{2.0}{2.5}{\ard};\darr{2.0}{3.0}{\ad};\darr{2.0}{3.5}{\ald};\darr{2.0}{4.0}{\al};
\darr{2.5}{0.0}{\al};\darr{2.5}{0.5}{\alu};\darr{2.5}{1.0}{\au};\darr{2.5}{1.5}{\aru};\darr{2.5}{2.0}{\ar};\darr{2.5}{2.5}{\ard};\darr{2.5}{3.0}{\ad};\darr{2.5}{3.5}{\ald};\darr{2.5}{4.0}{\al};
\darr{3.0}{0.0}{\al};\darr{3.0}{0.5}{\alu};\darr{3.0}{1.0}{\au};\darr{3.0}{1.5}{\aru};\darr{3.0}{2.0}{\ar};\darr{3.0}{2.5}{\ard};\darr{3.0}{3.0}{\ad};\darr{3.0}{3.5}{\ald};\darr{3.0}{4.0}{\al};
\darr{3.5}{0.0}{\al};\darr{3.5}{0.5}{\alu};\darr{3.5}{1.0}{\au};\darr{3.5}{1.5}{\aru};\darr{3.5}{2.0}{\ar};\darr{3.5}{2.5}{\ard};\darr{3.5}{3.0}{\ad};\darr{3.5}{3.5}{\ald};\darr{3.5}{4.0}{\al};
\darr{4.0}{0.0}{\al};\darr{4.0}{0.5}{\alu};\darr{4.0}{1.0}{\au};\darr{4.0}{1.5}{\aru};\darr{4.0}{2.0}{\ar};\darr{4.0}{2.5}{\ard};\darr{4.0}{3.0}{\ad};\darr{4.0}{3.5}{\ald};\darr{4.0}{4.0}{\al};

        \end{scope}

        \begin{scope}[xshift= 0 * \xstep, yshift = 4 * \ystep]       
            \begin{scope}[yshift=2mm, xshift=-3mm]
                \draw[very thick, ->] (-1.5, 2.5) -> (-1.5, 2.1);
                \draw[very thick, ->] (-1.4, 2.0) -> (-1.0, 2.0);
                \node[fill=black, circle, inner sep=0.7mm] at (-1.5, 2) {};

                \node[xshift=3mm] at (-0.7, 2) {$\Rightarrow$};
            \end{scope}

            \draw[step=.5cm,style=help lines] (0, 0) grid (4.5, 4.5);

            \darr{0.0}{0.0}{\al};\darr{0.0}{0.5}{\al};\darr{0.0}{1.0}{\al};\darr{0.0}{1.5}{\al};\darr{0.0}{2.0}{\al};\darr{0.0}{2.5}{\al};\darr{0.0}{3.0}{\al};\darr{0.0}{3.5}{\al};\darr{0.0}{4.0}{\al};
\darr{0.5}{0.0}{\al};\darr{0.5}{0.5}{\ald};\darr{0.5}{1.0}{\ald};\darr{0.5}{1.5}{\ald};\darr{0.5}{2.0}{\ald};\darr{0.5}{2.5}{\ald};\darr{0.5}{3.0}{\ald};\darr{0.5}{3.5}{\ald};\darr{0.5}{4.0}{\ald};
\darr{1.0}{0.0}{\al};\darr{1.0}{0.5}{\ald};\darr{1.0}{1.0}{\ad};\darr{1.0}{1.5}{\ad};\darr{1.0}{2.0}{\ad};\darr{1.0}{2.5}{\ad};\darr{1.0}{3.0}{\ad};\darr{1.0}{3.5}{\ad};\darr{1.0}{4.0}{\ad};
\darr{1.5}{0.0}{\al};\darr{1.5}{0.5}{\ald};\darr{1.5}{1.0}{\ad};\darr{1.5}{1.5}{\ard};\darr{1.5}{2.0}{\ard};\darr{1.5}{2.5}{\ard};\darr{1.5}{3.0}{\ard};\darr{1.5}{3.5}{\ard};\darr{1.5}{4.0}{\ard};
\darr{2.0}{0.0}{\al};\darr{2.0}{0.5}{\ald};\darr{2.0}{1.0}{\ad};\darr{2.0}{1.5}{\ard};\darr{2.0}{2.0}{\ar};\darr{2.0}{2.5}{\ar};\darr{2.0}{3.0}{\ar};\darr{2.0}{3.5}{\ar};\darr{2.0}{4.0}{\ar};
\darr{2.5}{0.0}{\al};\darr{2.5}{0.5}{\ald};\darr{2.5}{1.0}{\ad};\darr{2.5}{1.5}{\ard};\darr{2.5}{2.0}{\ar};\darr{2.5}{2.5}{\aru};\darr{2.5}{3.0}{\aru};\darr{2.5}{3.5}{\aru};\darr{2.5}{4.0}{\aru};
\darr{3.0}{0.0}{\al};\darr{3.0}{0.5}{\ald};\darr{3.0}{1.0}{\ad};\darr{3.0}{1.5}{\ard};\darr{3.0}{2.0}{\ar};\darr{3.0}{2.5}{\aru};\darr{3.0}{3.0}{\au};\darr{3.0}{3.5}{\au};\darr{3.0}{4.0}{\au};
\darr{3.5}{0.0}{\al};\darr{3.5}{0.5}{\ald};\darr{3.5}{1.0}{\ad};\darr{3.5}{1.5}{\ard};\darr{3.5}{2.0}{\ar};\darr{3.5}{2.5}{\aru};\darr{3.5}{3.0}{\au};\darr{3.5}{3.5}{\alu};\darr{3.5}{4.0}{\alu};
\darr{4.0}{0.0}{\al};\darr{4.0}{0.5}{\ald};\darr{4.0}{1.0}{\ad};\darr{4.0}{1.5}{\ard};\darr{4.0}{2.0}{\ar};\darr{4.0}{2.5}{\aru};\darr{4.0}{3.0}{\au};\darr{4.0}{3.5}{\alu};\darr{4.0}{4.0}{\al};

        \end{scope}

        \begin{scope}[xshift= 1 * \xstep, yshift = 4 * \ystep]       
            \begin{scope}[yshift=2mm, xshift=-3mm]
                \draw[very thick, ->] (-2.0, 2.0) -> (-1.6, 2.0);
                \draw[very thick, ->] (-1.4, 2.0) -> (-1.0, 2.0);
                \node[fill=black, circle, inner sep=0.7mm] at (-1.5, 2) {};

                \node[xshift=3mm] at (-0.7, 2) {$\Rightarrow$};
            \end{scope}

            \draw[step=.5cm,style=help lines] (0, 0) grid (4.5, 4.5);

            \darr{0.0}{0.0}{\al};\darr{0.0}{0.5}{\al};\darr{0.0}{1.0}{\al};\darr{0.0}{1.5}{\al};\darr{0.0}{2.0}{\al};\darr{0.0}{2.5}{\al};\darr{0.0}{3.0}{\al};\darr{0.0}{3.5}{\al};\darr{0.0}{4.0}{\al};
\darr{0.5}{0.0}{\ald};\darr{0.5}{0.5}{\ald};\darr{0.5}{1.0}{\ald};\darr{0.5}{1.5}{\ald};\darr{0.5}{2.0}{\ald};\darr{0.5}{2.5}{\ald};\darr{0.5}{3.0}{\ald};\darr{0.5}{3.5}{\ald};\darr{0.5}{4.0}{\ald};
\darr{1.0}{0.0}{\ad};\darr{1.0}{0.5}{\ad};\darr{1.0}{1.0}{\ad};\darr{1.0}{1.5}{\ad};\darr{1.0}{2.0}{\ad};\darr{1.0}{2.5}{\ad};\darr{1.0}{3.0}{\ad};\darr{1.0}{3.5}{\ad};\darr{1.0}{4.0}{\ad};
\darr{1.5}{0.0}{\ard};\darr{1.5}{0.5}{\ard};\darr{1.5}{1.0}{\ard};\darr{1.5}{1.5}{\ard};\darr{1.5}{2.0}{\ard};\darr{1.5}{2.5}{\ard};\darr{1.5}{3.0}{\ard};\darr{1.5}{3.5}{\ard};\darr{1.5}{4.0}{\ard};
\darr{2.0}{0.0}{\ar};\darr{2.0}{0.5}{\ar};\darr{2.0}{1.0}{\ar};\darr{2.0}{1.5}{\ar};\darr{2.0}{2.0}{\ar};\darr{2.0}{2.5}{\ar};\darr{2.0}{3.0}{\ar};\darr{2.0}{3.5}{\ar};\darr{2.0}{4.0}{\ar};
\darr{2.5}{0.0}{\aru};\darr{2.5}{0.5}{\aru};\darr{2.5}{1.0}{\aru};\darr{2.5}{1.5}{\aru};\darr{2.5}{2.0}{\aru};\darr{2.5}{2.5}{\aru};\darr{2.5}{3.0}{\aru};\darr{2.5}{3.5}{\aru};\darr{2.5}{4.0}{\aru};
\darr{3.0}{0.0}{\au};\darr{3.0}{0.5}{\au};\darr{3.0}{1.0}{\au};\darr{3.0}{1.5}{\au};\darr{3.0}{2.0}{\au};\darr{3.0}{2.5}{\au};\darr{3.0}{3.0}{\au};\darr{3.0}{3.5}{\au};\darr{3.0}{4.0}{\au};
\darr{3.5}{0.0}{\alu};\darr{3.5}{0.5}{\alu};\darr{3.5}{1.0}{\alu};\darr{3.5}{1.5}{\alu};\darr{3.5}{2.0}{\alu};\darr{3.5}{2.5}{\alu};\darr{3.5}{3.0}{\alu};\darr{3.5}{3.5}{\alu};\darr{3.5}{4.0}{\alu};
\darr{4.0}{0.0}{\al};\darr{4.0}{0.5}{\al};\darr{4.0}{1.0}{\al};\darr{4.0}{1.5}{\al};\darr{4.0}{2.0}{\al};\darr{4.0}{2.5}{\al};\darr{4.0}{3.0}{\al};\darr{4.0}{3.5}{\al};\darr{4.0}{4.0}{\al};

        \end{scope}

        \begin{scope}[xshift= 2 * \xstep, yshift = 4 * \ystep]       
            \begin{scope}[yshift=2mm, xshift=-3mm]
                \draw[very thick, ->] (-2.0, 2.0) -> (-1.6, 2.0);
                \draw[very thick, ->] (-1.5, 2.1) -> (-1.5, 2.5);
                \node[fill=black, circle, inner sep=0.7mm] at (-1.5, 2) {};

                \node[xshift=3mm] at (-0.7, 2) {$\Rightarrow$};
            \end{scope}

            \draw[step=.5cm,style=help lines] (0, 0) grid (4.5, 4.5);

            \darr{0.0}{0.0}{\al};\darr{0.0}{0.5}{\al};\darr{0.0}{1.0}{\al};\darr{0.0}{1.5}{\al};\darr{0.0}{2.0}{\al};\darr{0.0}{2.5}{\al};\darr{0.0}{3.0}{\al};\darr{0.0}{3.5}{\al};\darr{0.0}{4.0}{\al};
\darr{0.5}{0.0}{\ald};\darr{0.5}{0.5}{\ald};\darr{0.5}{1.0}{\ald};\darr{0.5}{1.5}{\ald};\darr{0.5}{2.0}{\ald};\darr{0.5}{2.5}{\ald};\darr{0.5}{3.0}{\ald};\darr{0.5}{3.5}{\ald};\darr{0.5}{4.0}{\al};
\darr{1.0}{0.0}{\ad};\darr{1.0}{0.5}{\ad};\darr{1.0}{1.0}{\ad};\darr{1.0}{1.5}{\ad};\darr{1.0}{2.0}{\ad};\darr{1.0}{2.5}{\ad};\darr{1.0}{3.0}{\ad};\darr{1.0}{3.5}{\ald};\darr{1.0}{4.0}{\al};
\darr{1.5}{0.0}{\ard};\darr{1.5}{0.5}{\ard};\darr{1.5}{1.0}{\ard};\darr{1.5}{1.5}{\ard};\darr{1.5}{2.0}{\ard};\darr{1.5}{2.5}{\ard};\darr{1.5}{3.0}{\ad};\darr{1.5}{3.5}{\ald};\darr{1.5}{4.0}{\al};
\darr{2.0}{0.0}{\ar};\darr{2.0}{0.5}{\ar};\darr{2.0}{1.0}{\ar};\darr{2.0}{1.5}{\ar};\darr{2.0}{2.0}{\ar};\darr{2.0}{2.5}{\ard};\darr{2.0}{3.0}{\ad};\darr{2.0}{3.5}{\ald};\darr{2.0}{4.0}{\al};
\darr{2.5}{0.0}{\aru};\darr{2.5}{0.5}{\aru};\darr{2.5}{1.0}{\aru};\darr{2.5}{1.5}{\aru};\darr{2.5}{2.0}{\ar};\darr{2.5}{2.5}{\ard};\darr{2.5}{3.0}{\ad};\darr{2.5}{3.5}{\ald};\darr{2.5}{4.0}{\al};
\darr{3.0}{0.0}{\au};\darr{3.0}{0.5}{\au};\darr{3.0}{1.0}{\au};\darr{3.0}{1.5}{\aru};\darr{3.0}{2.0}{\ar};\darr{3.0}{2.5}{\ard};\darr{3.0}{3.0}{\ad};\darr{3.0}{3.5}{\ald};\darr{3.0}{4.0}{\al};
\darr{3.5}{0.0}{\alu};\darr{3.5}{0.5}{\alu};\darr{3.5}{1.0}{\au};\darr{3.5}{1.5}{\aru};\darr{3.5}{2.0}{\ar};\darr{3.5}{2.5}{\ard};\darr{3.5}{3.0}{\ad};\darr{3.5}{3.5}{\ald};\darr{3.5}{4.0}{\al};
\darr{4.0}{0.0}{\al};\darr{4.0}{0.5}{\alu};\darr{4.0}{1.0}{\au};\darr{4.0}{1.5}{\aru};\darr{4.0}{2.0}{\ar};\darr{4.0}{2.5}{\ard};\darr{4.0}{3.0}{\ad};\darr{4.0}{3.5}{\ald};\darr{4.0}{4.0}{\al};

        \end{scope}

        \begin{scope}[xshift= 0 * \xstep, yshift = 3 * \ystep]       
            \begin{scope}[yshift=2mm, xshift=-3mm]
                \draw[very thick, ->] (-1.5, 1.5) -> (-1.5, 1.9);
                \draw[very thick, ->] (-1.4, 2.0) -> (-1.0, 2.0);
                \node[fill=black, circle, inner sep=0.7mm] at (-1.5, 2) {};

                \node[xshift=3mm] at (-0.7, 2) {$\Rightarrow$};
            \end{scope}

            \draw[step=.5cm,style=help lines] (0, 0) grid (4.5, 4.5);

            \darr{0.0}{0.0}{\al};\darr{0.0}{0.5}{\alu};\darr{0.0}{1.0}{\au};\darr{0.0}{1.5}{\aru};\darr{0.0}{2.0}{\ar};\darr{0.0}{2.5}{\ard};\darr{0.0}{3.0}{\ad};\darr{0.0}{3.5}{\ald};\darr{0.0}{4.0}{\al};
\darr{0.5}{0.0}{\al};\darr{0.5}{0.5}{\alu};\darr{0.5}{1.0}{\au};\darr{0.5}{1.5}{\aru};\darr{0.5}{2.0}{\ar};\darr{0.5}{2.5}{\ard};\darr{0.5}{3.0}{\ad};\darr{0.5}{3.5}{\ald};\darr{0.5}{4.0}{\ald};
\darr{1.0}{0.0}{\al};\darr{1.0}{0.5}{\alu};\darr{1.0}{1.0}{\au};\darr{1.0}{1.5}{\aru};\darr{1.0}{2.0}{\ar};\darr{1.0}{2.5}{\ard};\darr{1.0}{3.0}{\ad};\darr{1.0}{3.5}{\ad};\darr{1.0}{4.0}{\ad};
\darr{1.5}{0.0}{\al};\darr{1.5}{0.5}{\alu};\darr{1.5}{1.0}{\au};\darr{1.5}{1.5}{\aru};\darr{1.5}{2.0}{\ar};\darr{1.5}{2.5}{\ard};\darr{1.5}{3.0}{\ard};\darr{1.5}{3.5}{\ard};\darr{1.5}{4.0}{\ard};
\darr{2.0}{0.0}{\al};\darr{2.0}{0.5}{\alu};\darr{2.0}{1.0}{\au};\darr{2.0}{1.5}{\aru};\darr{2.0}{2.0}{\ar};\darr{2.0}{2.5}{\ar};\darr{2.0}{3.0}{\ar};\darr{2.0}{3.5}{\ar};\darr{2.0}{4.0}{\ar};
\darr{2.5}{0.0}{\al};\darr{2.5}{0.5}{\alu};\darr{2.5}{1.0}{\au};\darr{2.5}{1.5}{\aru};\darr{2.5}{2.0}{\aru};\darr{2.5}{2.5}{\aru};\darr{2.5}{3.0}{\aru};\darr{2.5}{3.5}{\aru};\darr{2.5}{4.0}{\aru};
\darr{3.0}{0.0}{\al};\darr{3.0}{0.5}{\alu};\darr{3.0}{1.0}{\au};\darr{3.0}{1.5}{\au};\darr{3.0}{2.0}{\au};\darr{3.0}{2.5}{\au};\darr{3.0}{3.0}{\au};\darr{3.0}{3.5}{\au};\darr{3.0}{4.0}{\au};
\darr{3.5}{0.0}{\al};\darr{3.5}{0.5}{\alu};\darr{3.5}{1.0}{\alu};\darr{3.5}{1.5}{\alu};\darr{3.5}{2.0}{\alu};\darr{3.5}{2.5}{\alu};\darr{3.5}{3.0}{\alu};\darr{3.5}{3.5}{\alu};\darr{3.5}{4.0}{\alu};
\darr{4.0}{0.0}{\al};\darr{4.0}{0.5}{\al};\darr{4.0}{1.0}{\al};\darr{4.0}{1.5}{\al};\darr{4.0}{2.0}{\al};\darr{4.0}{2.5}{\al};\darr{4.0}{3.0}{\al};\darr{4.0}{3.5}{\al};\darr{4.0}{4.0}{\al};

        \end{scope}

        \begin{scope}[xshift= 1 * \xstep, yshift = 3 * \ystep]       
            \begin{scope}[yshift=2mm, xshift=-3mm]
                \draw[very thick, ->] (-2.0, 2.0) -> (-1.6, 2.0);
                \draw[very thick, ->] (-1.5, 1.9) -> (-1.5, 1.5);
                \node[fill=black, circle, inner sep=0.7mm] at (-1.5, 2) {};

                \node[xshift=3mm] at (-0.7, 2) {$\Rightarrow$};
            \end{scope}

            \draw[step=.5cm,style=help lines] (0, 0) grid (4.5, 4.5);

            \darr{0.0}{0.0}{\al};\darr{0.0}{0.5}{\ald};\darr{0.0}{1.0}{\ad};\darr{0.0}{1.5}{\ard};\darr{0.0}{2.0}{\ar};\darr{0.0}{2.5}{\aru};\darr{0.0}{3.0}{\au};\darr{0.0}{3.5}{\alu};\darr{0.0}{4.0}{\al};
\darr{0.5}{0.0}{\ald};\darr{0.5}{0.5}{\ald};\darr{0.5}{1.0}{\ad};\darr{0.5}{1.5}{\ard};\darr{0.5}{2.0}{\ar};\darr{0.5}{2.5}{\aru};\darr{0.5}{3.0}{\au};\darr{0.5}{3.5}{\alu};\darr{0.5}{4.0}{\al};
\darr{1.0}{0.0}{\ad};\darr{1.0}{0.5}{\ad};\darr{1.0}{1.0}{\ad};\darr{1.0}{1.5}{\ard};\darr{1.0}{2.0}{\ar};\darr{1.0}{2.5}{\aru};\darr{1.0}{3.0}{\au};\darr{1.0}{3.5}{\alu};\darr{1.0}{4.0}{\al};
\darr{1.5}{0.0}{\ard};\darr{1.5}{0.5}{\ard};\darr{1.5}{1.0}{\ard};\darr{1.5}{1.5}{\ard};\darr{1.5}{2.0}{\ar};\darr{1.5}{2.5}{\aru};\darr{1.5}{3.0}{\au};\darr{1.5}{3.5}{\alu};\darr{1.5}{4.0}{\al};
\darr{2.0}{0.0}{\ar};\darr{2.0}{0.5}{\ar};\darr{2.0}{1.0}{\ar};\darr{2.0}{1.5}{\ar};\darr{2.0}{2.0}{\ar};\darr{2.0}{2.5}{\aru};\darr{2.0}{3.0}{\au};\darr{2.0}{3.5}{\alu};\darr{2.0}{4.0}{\al};
\darr{2.5}{0.0}{\aru};\darr{2.5}{0.5}{\aru};\darr{2.5}{1.0}{\aru};\darr{2.5}{1.5}{\aru};\darr{2.5}{2.0}{\aru};\darr{2.5}{2.5}{\aru};\darr{2.5}{3.0}{\au};\darr{2.5}{3.5}{\alu};\darr{2.5}{4.0}{\al};
\darr{3.0}{0.0}{\au};\darr{3.0}{0.5}{\au};\darr{3.0}{1.0}{\au};\darr{3.0}{1.5}{\au};\darr{3.0}{2.0}{\au};\darr{3.0}{2.5}{\au};\darr{3.0}{3.0}{\au};\darr{3.0}{3.5}{\alu};\darr{3.0}{4.0}{\al};
\darr{3.5}{0.0}{\alu};\darr{3.5}{0.5}{\alu};\darr{3.5}{1.0}{\alu};\darr{3.5}{1.5}{\alu};\darr{3.5}{2.0}{\alu};\darr{3.5}{2.5}{\alu};\darr{3.5}{3.0}{\alu};\darr{3.5}{3.5}{\alu};\darr{3.5}{4.0}{\al};
\darr{4.0}{0.0}{\al};\darr{4.0}{0.5}{\al};\darr{4.0}{1.0}{\al};\darr{4.0}{1.5}{\al};\darr{4.0}{2.0}{\al};\darr{4.0}{2.5}{\al};\darr{4.0}{3.0}{\al};\darr{4.0}{3.5}{\al};\darr{4.0}{4.0}{\al};

        \end{scope}

        \begin{scope}[xshift= 2 * \xstep, yshift = 3 * \ystep]       
            \begin{scope}[yshift=2mm, xshift=-3mm]
                \draw[very thick, ->] (-1.5, 1.5) -> (-1.5, 1.9);
                
                \node[fill=black, circle, inner sep=0.7mm] at (-1.5, 2) {};

                \node[xshift=3mm] at (-0.7, 2) {$\Rightarrow$};
            \end{scope}

            \draw[step=.5cm,style=help lines] (0, 0) grid (4.5, 4.5);

            \darr{0.0}{0.0}{\al};\darr{0.0}{0.5}{\alu};\darr{0.0}{1.0}{\au};\darr{0.0}{1.5}{\aru};\darr{0.0}{2.0}{\ar};\darr{0.0}{2.5}{\ard};\darr{0.0}{3.0}{\ad};\darr{0.0}{3.5}{\ald};\darr{0.0}{4.0}{\al};
\darr{0.5}{0.0}{\al};\darr{0.5}{0.5}{\alu};\darr{0.5}{1.0}{\au};\darr{0.5}{1.5}{\aru};\darr{0.5}{2.0}{\ar};\darr{0.5}{2.5}{\aru};\darr{0.5}{3.0}{\au};\darr{0.5}{3.5}{\alu};\darr{0.5}{4.0}{\al};
\darr{1.0}{0.0}{\al};\darr{1.0}{0.5}{\alu};\darr{1.0}{1.0}{\au};\darr{1.0}{1.5}{\aru};\darr{1.0}{2.0}{\ar};\darr{1.0}{2.5}{\aru};\darr{1.0}{3.0}{\au};\darr{1.0}{3.5}{\alu};\darr{1.0}{4.0}{\al};
\darr{1.5}{0.0}{\al};\darr{1.5}{0.5}{\alu};\darr{1.5}{1.0}{\au};\darr{1.5}{1.5}{\aru};\darr{1.5}{2.0}{\ar};\darr{1.5}{2.5}{\aru};\darr{1.5}{3.0}{\au};\darr{1.5}{3.5}{\alu};\darr{1.5}{4.0}{\al};
\darr{2.0}{0.0}{\al};\darr{2.0}{0.5}{\alu};\darr{2.0}{1.0}{\au};\darr{2.0}{1.5}{\aru};\darr{2.0}{2.0}{\ar};\darr{2.0}{2.5}{\aru};\darr{2.0}{3.0}{\au};\darr{2.0}{3.5}{\alu};\darr{2.0}{4.0}{\al};
\darr{2.5}{0.0}{\al};\darr{2.5}{0.5}{\alu};\darr{2.5}{1.0}{\au};\darr{2.5}{1.5}{\aru};\darr{2.5}{2.0}{\aru};\darr{2.5}{2.5}{\aru};\darr{2.5}{3.0}{\au};\darr{2.5}{3.5}{\alu};\darr{2.5}{4.0}{\al};
\darr{3.0}{0.0}{\al};\darr{3.0}{0.5}{\alu};\darr{3.0}{1.0}{\au};\darr{3.0}{1.5}{\au};\darr{3.0}{2.0}{\au};\darr{3.0}{2.5}{\au};\darr{3.0}{3.0}{\au};\darr{3.0}{3.5}{\alu};\darr{3.0}{4.0}{\al};
\darr{3.5}{0.0}{\al};\darr{3.5}{0.5}{\alu};\darr{3.5}{1.0}{\alu};\darr{3.5}{1.5}{\alu};\darr{3.5}{2.0}{\alu};\darr{3.5}{2.5}{\alu};\darr{3.5}{3.0}{\alu};\darr{3.5}{3.5}{\alu};\darr{3.5}{4.0}{\al};
\darr{4.0}{0.0}{\al};\darr{4.0}{0.5}{\al};\darr{4.0}{1.0}{\al};\darr{4.0}{1.5}{\al};\darr{4.0}{2.0}{\al};\darr{4.0}{2.5}{\al};\darr{4.0}{3.0}{\al};\darr{4.0}{3.5}{\al};\darr{4.0}{4.0}{\al};
\draw[thick, color=red] (3.25, 0.5) ellipse (0.8 and 0.6);

        \end{scope}

        \begin{scope}[xshift= 0 * \xstep, yshift = 2 * \ystep]       
            \begin{scope}[yshift=2mm, xshift=-3mm]
                \draw[very thick, ->] (-1.0, 2.0) -> (-1.4, 2.0);
                \draw[very thick, ->] (-1.5, 2.1) -> (-1.5, 2.5);
                \node[fill=black, circle, inner sep=0.7mm] at (-1.5, 2) {};

                \node[xshift=3mm] at (-0.7, 2) {$\Rightarrow$};
            \end{scope}

            \draw[step=.5cm,style=help lines] (0, 0) grid (4.5, 4.5);

            \darr{0.0}{0.0}{\al};\darr{0.0}{0.5}{\al};\darr{0.0}{1.0}{\al};\darr{0.0}{1.5}{\al};\darr{0.0}{2.0}{\al};\darr{0.0}{2.5}{\al};\darr{0.0}{3.0}{\al};\darr{0.0}{3.5}{\al};\darr{0.0}{4.0}{\al};
\darr{0.5}{0.0}{\al};\darr{0.5}{0.5}{\alu};\darr{0.5}{1.0}{\alu};\darr{0.5}{1.5}{\alu};\darr{0.5}{2.0}{\alu};\darr{0.5}{2.5}{\alu};\darr{0.5}{3.0}{\alu};\darr{0.5}{3.5}{\alu};\darr{0.5}{4.0}{\alu};
\darr{1.0}{0.0}{\al};\darr{1.0}{0.5}{\alu};\darr{1.0}{1.0}{\au};\darr{1.0}{1.5}{\au};\darr{1.0}{2.0}{\au};\darr{1.0}{2.5}{\au};\darr{1.0}{3.0}{\au};\darr{1.0}{3.5}{\au};\darr{1.0}{4.0}{\au};
\darr{1.5}{0.0}{\al};\darr{1.5}{0.5}{\alu};\darr{1.5}{1.0}{\au};\darr{1.5}{1.5}{\aru};\darr{1.5}{2.0}{\aru};\darr{1.5}{2.5}{\aru};\darr{1.5}{3.0}{\aru};\darr{1.5}{3.5}{\aru};\darr{1.5}{4.0}{\aru};
\darr{2.0}{0.0}{\al};\darr{2.0}{0.5}{\alu};\darr{2.0}{1.0}{\au};\darr{2.0}{1.5}{\aru};\darr{2.0}{2.0}{\ar};\darr{2.0}{2.5}{\ar};\darr{2.0}{3.0}{\ar};\darr{2.0}{3.5}{\ar};\darr{2.0}{4.0}{\ar};
\darr{2.5}{0.0}{\al};\darr{2.5}{0.5}{\alu};\darr{2.5}{1.0}{\au};\darr{2.5}{1.5}{\aru};\darr{2.5}{2.0}{\ar};\darr{2.5}{2.5}{\ard};\darr{2.5}{3.0}{\ard};\darr{2.5}{3.5}{\ard};\darr{2.5}{4.0}{\ard};
\darr{3.0}{0.0}{\al};\darr{3.0}{0.5}{\alu};\darr{3.0}{1.0}{\au};\darr{3.0}{1.5}{\aru};\darr{3.0}{2.0}{\ar};\darr{3.0}{2.5}{\ard};\darr{3.0}{3.0}{\ad};\darr{3.0}{3.5}{\ad};\darr{3.0}{4.0}{\ad};
\darr{3.5}{0.0}{\al};\darr{3.5}{0.5}{\alu};\darr{3.5}{1.0}{\au};\darr{3.5}{1.5}{\aru};\darr{3.5}{2.0}{\ar};\darr{3.5}{2.5}{\ard};\darr{3.5}{3.0}{\ad};\darr{3.5}{3.5}{\ald};\darr{3.5}{4.0}{\ald};
\darr{4.0}{0.0}{\al};\darr{4.0}{0.5}{\alu};\darr{4.0}{1.0}{\au};\darr{4.0}{1.5}{\aru};\darr{4.0}{2.0}{\ar};\darr{4.0}{2.5}{\ard};\darr{4.0}{3.0}{\ad};\darr{4.0}{3.5}{\ald};\darr{4.0}{4.0}{\al};

        \end{scope}

        \begin{scope}[xshift= 1 * \xstep, yshift = 2 * \ystep]       
            \begin{scope}[yshift=2mm, xshift=-3mm]
                \draw[very thick, ->] (-1.5, 2.5) -> (-1.5, 2.1);
                \draw[very thick, ->] (-1.6, 2.0) -> (-2.0, 2.0);
                \node[fill=black, circle, inner sep=0.7mm] at (-1.5, 2) {};

                \node[xshift=3mm] at (-0.7, 2) {$\Rightarrow$};
            \end{scope}

            \draw[step=.5cm,style=help lines] (0, 0) grid (4.5, 4.5);

            \darr{0.0}{0.0}{\al};\darr{0.0}{0.5}{\al};\darr{0.0}{1.0}{\al};\darr{0.0}{1.5}{\al};\darr{0.0}{2.0}{\al};\darr{0.0}{2.5}{\al};\darr{0.0}{3.0}{\al};\darr{0.0}{3.5}{\al};\darr{0.0}{4.0}{\al};
\darr{0.5}{0.0}{\alu};\darr{0.5}{0.5}{\alu};\darr{0.5}{1.0}{\alu};\darr{0.5}{1.5}{\alu};\darr{0.5}{2.0}{\alu};\darr{0.5}{2.5}{\alu};\darr{0.5}{3.0}{\alu};\darr{0.5}{3.5}{\alu};\darr{0.5}{4.0}{\al};
\darr{1.0}{0.0}{\au};\darr{1.0}{0.5}{\au};\darr{1.0}{1.0}{\au};\darr{1.0}{1.5}{\au};\darr{1.0}{2.0}{\au};\darr{1.0}{2.5}{\au};\darr{1.0}{3.0}{\au};\darr{1.0}{3.5}{\alu};\darr{1.0}{4.0}{\al};
\darr{1.5}{0.0}{\aru};\darr{1.5}{0.5}{\aru};\darr{1.5}{1.0}{\aru};\darr{1.5}{1.5}{\aru};\darr{1.5}{2.0}{\aru};\darr{1.5}{2.5}{\aru};\darr{1.5}{3.0}{\au};\darr{1.5}{3.5}{\alu};\darr{1.5}{4.0}{\al};
\darr{2.0}{0.0}{\ar};\darr{2.0}{0.5}{\ar};\darr{2.0}{1.0}{\ar};\darr{2.0}{1.5}{\ar};\darr{2.0}{2.0}{\ar};\darr{2.0}{2.5}{\aru};\darr{2.0}{3.0}{\au};\darr{2.0}{3.5}{\alu};\darr{2.0}{4.0}{\al};
\darr{2.5}{0.0}{\ard};\darr{2.5}{0.5}{\ard};\darr{2.5}{1.0}{\ard};\darr{2.5}{1.5}{\ard};\darr{2.5}{2.0}{\ar};\darr{2.5}{2.5}{\aru};\darr{2.5}{3.0}{\au};\darr{2.5}{3.5}{\alu};\darr{2.5}{4.0}{\al};
\darr{3.0}{0.0}{\ad};\darr{3.0}{0.5}{\ad};\darr{3.0}{1.0}{\ad};\darr{3.0}{1.5}{\ard};\darr{3.0}{2.0}{\ar};\darr{3.0}{2.5}{\aru};\darr{3.0}{3.0}{\au};\darr{3.0}{3.5}{\alu};\darr{3.0}{4.0}{\al};
\darr{3.5}{0.0}{\ald};\darr{3.5}{0.5}{\ald};\darr{3.5}{1.0}{\ad};\darr{3.5}{1.5}{\ard};\darr{3.5}{2.0}{\ar};\darr{3.5}{2.5}{\aru};\darr{3.5}{3.0}{\au};\darr{3.5}{3.5}{\alu};\darr{3.5}{4.0}{\al};
\darr{4.0}{0.0}{\al};\darr{4.0}{0.5}{\ald};\darr{4.0}{1.0}{\ad};\darr{4.0}{1.5}{\ard};\darr{4.0}{2.0}{\ar};\darr{4.0}{2.5}{\aru};\darr{4.0}{3.0}{\au};\darr{4.0}{3.5}{\alu};\darr{4.0}{4.0}{\al};

        \end{scope}

        \begin{scope}[xshift= 2 * \xstep, yshift = 2 * \ystep]       
            \begin{scope}[yshift=2mm, xshift=-3mm]
                
                \draw[very thick, ->] (-1.5, 2.1) -> (-1.5, 2.5);
                \node[fill=black, circle, inner sep=0.7mm] at (-1.5, 2) {};

                \node[xshift=3mm] at (-0.7, 2) {$\Rightarrow$};
            \end{scope}

            \draw[step=.5cm,style=help lines] (0, 0) grid (4.5, 4.5);

            \darr{0.0}{0.0}{\al};\darr{0.0}{0.5}{\al};\darr{0.0}{1.0}{\al};\darr{0.0}{1.5}{\al};\darr{0.0}{2.0}{\al};\darr{0.0}{2.5}{\al};\darr{0.0}{3.0}{\al};\darr{0.0}{3.5}{\al};\darr{0.0}{4.0}{\al};
\darr{0.5}{0.0}{\al};\darr{0.5}{0.5}{\alu};\darr{0.5}{1.0}{\alu};\darr{0.5}{1.5}{\alu};\darr{0.5}{2.0}{\alu};\darr{0.5}{2.5}{\alu};\darr{0.5}{3.0}{\alu};\darr{0.5}{3.5}{\alu};\darr{0.5}{4.0}{\al};
\darr{1.0}{0.0}{\al};\darr{1.0}{0.5}{\alu};\darr{1.0}{1.0}{\au};\darr{1.0}{1.5}{\au};\darr{1.0}{2.0}{\au};\darr{1.0}{2.5}{\au};\darr{1.0}{3.0}{\au};\darr{1.0}{3.5}{\alu};\darr{1.0}{4.0}{\al};
\darr{1.5}{0.0}{\al};\darr{1.5}{0.5}{\alu};\darr{1.5}{1.0}{\au};\darr{1.5}{1.5}{\aru};\darr{1.5}{2.0}{\aru};\darr{1.5}{2.5}{\aru};\darr{1.5}{3.0}{\au};\darr{1.5}{3.5}{\alu};\darr{1.5}{4.0}{\al};
\darr{2.0}{0.0}{\al};\darr{2.0}{0.5}{\alu};\darr{2.0}{1.0}{\au};\darr{2.0}{1.5}{\aru};\darr{2.0}{2.0}{\ar};\darr{2.0}{2.5}{\aru};\darr{2.0}{3.0}{\au};\darr{2.0}{3.5}{\alu};\darr{2.0}{4.0}{\al};
\darr{2.5}{0.0}{\al};\darr{2.5}{0.5}{\alu};\darr{2.5}{1.0}{\au};\darr{2.5}{1.5}{\aru};\darr{2.5}{2.0}{\ar};\darr{2.5}{2.5}{\aru};\darr{2.5}{3.0}{\au};\darr{2.5}{3.5}{\alu};\darr{2.5}{4.0}{\al};
\darr{3.0}{0.0}{\al};\darr{3.0}{0.5}{\alu};\darr{3.0}{1.0}{\au};\darr{3.0}{1.5}{\aru};\darr{3.0}{2.0}{\ar};\darr{3.0}{2.5}{\aru};\darr{3.0}{3.0}{\au};\darr{3.0}{3.5}{\alu};\darr{3.0}{4.0}{\al};
\darr{3.5}{0.0}{\al};\darr{3.5}{0.5}{\alu};\darr{3.5}{1.0}{\au};\darr{3.5}{1.5}{\aru};\darr{3.5}{2.0}{\ar};\darr{3.5}{2.5}{\aru};\darr{3.5}{3.0}{\au};\darr{3.5}{3.5}{\alu};\darr{3.5}{4.0}{\al};
\darr{4.0}{0.0}{\al};\darr{4.0}{0.5}{\alu};\darr{4.0}{1.0}{\au};\darr{4.0}{1.5}{\aru};\darr{4.0}{2.0}{\ar};\darr{4.0}{2.5}{\ard};\darr{4.0}{3.0}{\ad};\darr{4.0}{3.5}{\ald};\darr{4.0}{4.0}{\al};
\draw[thick, color=red] (3.25, 4.0) ellipse (0.8 and 0.6);

        \end{scope}

        \begin{scope}[xshift= 0 * \xstep, yshift = 1 * \ystep]       
            \begin{scope}[yshift=2mm, xshift=-3mm]
                \draw[very thick, ->] (-1.0, 2.0) -> (-1.4, 2.0);
                
                \node[fill=black, circle, inner sep=0.7mm] at (-1.5, 2) {};

                \node[xshift=3mm] at (-0.7, 2) {$\Rightarrow$};
            \end{scope}

            \draw[step=.5cm,style=help lines] (0, 0) grid (4.5, 4.5);

            \darr{0.0}{0.0}{\al};\darr{0.0}{0.5}{\al};\darr{0.0}{1.0}{\al};\darr{0.0}{1.5}{\al};\darr{0.0}{2.0}{\al};\darr{0.0}{2.5}{\al};\darr{0.0}{3.0}{\al};\darr{0.0}{3.5}{\al};\darr{0.0}{4.0}{\al};
\darr{0.5}{0.0}{\al};\darr{0.5}{0.5}{\alu};\darr{0.5}{1.0}{\alu};\darr{0.5}{1.5}{\alu};\darr{0.5}{2.0}{\alu};\darr{0.5}{2.5}{\alu};\darr{0.5}{3.0}{\alu};\darr{0.5}{3.5}{\alu};\darr{0.5}{4.0}{\alu};
\darr{1.0}{0.0}{\al};\darr{1.0}{0.5}{\alu};\darr{1.0}{1.0}{\au};\darr{1.0}{1.5}{\au};\darr{1.0}{2.0}{\au};\darr{1.0}{2.5}{\au};\darr{1.0}{3.0}{\au};\darr{1.0}{3.5}{\au};\darr{1.0}{4.0}{\au};
\darr{1.5}{0.0}{\al};\darr{1.5}{0.5}{\alu};\darr{1.5}{1.0}{\au};\darr{1.5}{1.5}{\aru};\darr{1.5}{2.0}{\aru};\darr{1.5}{2.5}{\aru};\darr{1.5}{3.0}{\aru};\darr{1.5}{3.5}{\aru};\darr{1.5}{4.0}{\aru};
\darr{2.0}{0.0}{\al};\darr{2.0}{0.5}{\alu};\darr{2.0}{1.0}{\au};\darr{2.0}{1.5}{\aru};\darr{2.0}{2.0}{\ar};\darr{2.0}{2.5}{\ar};\darr{2.0}{3.0}{\ar};\darr{2.0}{3.5}{\ar};\darr{2.0}{4.0}{\ar};
\darr{2.5}{0.0}{\al};\darr{2.5}{0.5}{\alu};\darr{2.5}{1.0}{\au};\darr{2.5}{1.5}{\aru};\darr{2.5}{2.0}{\aru};\darr{2.5}{2.5}{\aru};\darr{2.5}{3.0}{\aru};\darr{2.5}{3.5}{\aru};\darr{2.5}{4.0}{\ard};
\darr{3.0}{0.0}{\al};\darr{3.0}{0.5}{\alu};\darr{3.0}{1.0}{\au};\darr{3.0}{1.5}{\au};\darr{3.0}{2.0}{\au};\darr{3.0}{2.5}{\au};\darr{3.0}{3.0}{\au};\darr{3.0}{3.5}{\au};\darr{3.0}{4.0}{\ad};
\darr{3.5}{0.0}{\al};\darr{3.5}{0.5}{\alu};\darr{3.5}{1.0}{\alu};\darr{3.5}{1.5}{\alu};\darr{3.5}{2.0}{\alu};\darr{3.5}{2.5}{\alu};\darr{3.5}{3.0}{\alu};\darr{3.5}{3.5}{\alu};\darr{3.5}{4.0}{\ald};
\darr{4.0}{0.0}{\al};\darr{4.0}{0.5}{\al};\darr{4.0}{1.0}{\al};\darr{4.0}{1.5}{\al};\darr{4.0}{2.0}{\al};\darr{4.0}{2.5}{\al};\darr{4.0}{3.0}{\al};\darr{4.0}{3.5}{\al};\darr{4.0}{4.0}{\al};
\draw[thick, color=red] (4.0, 3.25) ellipse (0.6 and 0.8);

        \end{scope}

        \begin{scope}[xshift= 1 * \xstep, yshift = 1 * \ystep]       
            \begin{scope}[yshift=2mm, xshift=-3mm]
                
                \draw[very thick, ->] (-1.6, 2.0) -> (-2.0, 2.0);
                \node[fill=black, circle, inner sep=0.7mm] at (-1.5, 2) {};

                \node[xshift=3mm] at (-0.7, 2) {$\Rightarrow$};
            \end{scope}

            \draw[step=.5cm,style=help lines] (0, 0) grid (4.5, 4.5);

            \darr{0.0}{0.0}{\al};\darr{0.0}{0.5}{\al};\darr{0.0}{1.0}{\al};\darr{0.0}{1.5}{\al};\darr{0.0}{2.0}{\al};\darr{0.0}{2.5}{\al};\darr{0.0}{3.0}{\al};\darr{0.0}{3.5}{\al};\darr{0.0}{4.0}{\al};
\darr{0.5}{0.0}{\alu};\darr{0.5}{0.5}{\alu};\darr{0.5}{1.0}{\alu};\darr{0.5}{1.5}{\alu};\darr{0.5}{2.0}{\alu};\darr{0.5}{2.5}{\alu};\darr{0.5}{3.0}{\alu};\darr{0.5}{3.5}{\alu};\darr{0.5}{4.0}{\al};
\darr{1.0}{0.0}{\au};\darr{1.0}{0.5}{\au};\darr{1.0}{1.0}{\au};\darr{1.0}{1.5}{\au};\darr{1.0}{2.0}{\au};\darr{1.0}{2.5}{\au};\darr{1.0}{3.0}{\au};\darr{1.0}{3.5}{\alu};\darr{1.0}{4.0}{\al};
\darr{1.5}{0.0}{\aru};\darr{1.5}{0.5}{\aru};\darr{1.5}{1.0}{\aru};\darr{1.5}{1.5}{\aru};\darr{1.5}{2.0}{\aru};\darr{1.5}{2.5}{\aru};\darr{1.5}{3.0}{\au};\darr{1.5}{3.5}{\alu};\darr{1.5}{4.0}{\al};
\darr{2.0}{0.0}{\ar};\darr{2.0}{0.5}{\ar};\darr{2.0}{1.0}{\ar};\darr{2.0}{1.5}{\ar};\darr{2.0}{2.0}{\ar};\darr{2.0}{2.5}{\aru};\darr{2.0}{3.0}{\au};\darr{2.0}{3.5}{\alu};\darr{2.0}{4.0}{\al};
\darr{2.5}{0.0}{\ard};\darr{2.5}{0.5}{\aru};\darr{2.5}{1.0}{\aru};\darr{2.5}{1.5}{\aru};\darr{2.5}{2.0}{\aru};\darr{2.5}{2.5}{\aru};\darr{2.5}{3.0}{\au};\darr{2.5}{3.5}{\alu};\darr{2.5}{4.0}{\al};
\darr{3.0}{0.0}{\ad};\darr{3.0}{0.5}{\au};\darr{3.0}{1.0}{\au};\darr{3.0}{1.5}{\au};\darr{3.0}{2.0}{\au};\darr{3.0}{2.5}{\au};\darr{3.0}{3.0}{\au};\darr{3.0}{3.5}{\alu};\darr{3.0}{4.0}{\al};
\darr{3.5}{0.0}{\ald};\darr{3.5}{0.5}{\alu};\darr{3.5}{1.0}{\alu};\darr{3.5}{1.5}{\alu};\darr{3.5}{2.0}{\alu};\darr{3.5}{2.5}{\alu};\darr{3.5}{3.0}{\alu};\darr{3.5}{3.5}{\alu};\darr{3.5}{4.0}{\al};
\darr{4.0}{0.0}{\al};\darr{4.0}{0.5}{\al};\darr{4.0}{1.0}{\al};\darr{4.0}{1.5}{\al};\darr{4.0}{2.0}{\al};\darr{4.0}{2.5}{\al};\darr{4.0}{3.0}{\al};\darr{4.0}{3.5}{\al};\darr{4.0}{4.0}{\al};
\draw[thick, color=red] (0.5, 3.25) ellipse (0.6 and 0.8);

        \end{scope}

        \begin{scope}[xshift= 2 * \xstep, yshift = 1 * \ystep]       
            \begin{scope}[yshift=2mm, xshift=-3mm]
                
                \draw[very thick, ->] (-1.5, 1.9) -> (-1.5, 1.5);
                \node[fill=black, circle, inner sep=0.7mm] at (-1.5, 2) {};

                \node[xshift=3mm] at (-0.7, 2) {$\Rightarrow$};
            \end{scope}

            \draw[step=.5cm,style=help lines] (0, 0) grid (4.5, 4.5);

            \darr{0.0}{0.0}{\al};\darr{0.0}{0.5}{\ald};\darr{0.0}{1.0}{\ad};\darr{0.0}{1.5}{\ard};\darr{0.0}{2.0}{\ar};\darr{0.0}{2.5}{\aru};\darr{0.0}{3.0}{\au};\darr{0.0}{3.5}{\alu};\darr{0.0}{4.0}{\al};
\darr{0.5}{0.0}{\al};\darr{0.5}{0.5}{\alu};\darr{0.5}{1.0}{\au};\darr{0.5}{1.5}{\aru};\darr{0.5}{2.0}{\ar};\darr{0.5}{2.5}{\aru};\darr{0.5}{3.0}{\au};\darr{0.5}{3.5}{\alu};\darr{0.5}{4.0}{\al};
\darr{1.0}{0.0}{\al};\darr{1.0}{0.5}{\alu};\darr{1.0}{1.0}{\au};\darr{1.0}{1.5}{\aru};\darr{1.0}{2.0}{\ar};\darr{1.0}{2.5}{\aru};\darr{1.0}{3.0}{\au};\darr{1.0}{3.5}{\alu};\darr{1.0}{4.0}{\al};
\darr{1.5}{0.0}{\al};\darr{1.5}{0.5}{\alu};\darr{1.5}{1.0}{\au};\darr{1.5}{1.5}{\aru};\darr{1.5}{2.0}{\ar};\darr{1.5}{2.5}{\aru};\darr{1.5}{3.0}{\au};\darr{1.5}{3.5}{\alu};\darr{1.5}{4.0}{\al};
\darr{2.0}{0.0}{\al};\darr{2.0}{0.5}{\alu};\darr{2.0}{1.0}{\au};\darr{2.0}{1.5}{\aru};\darr{2.0}{2.0}{\ar};\darr{2.0}{2.5}{\aru};\darr{2.0}{3.0}{\au};\darr{2.0}{3.5}{\alu};\darr{2.0}{4.0}{\al};
\darr{2.5}{0.0}{\al};\darr{2.5}{0.5}{\alu};\darr{2.5}{1.0}{\au};\darr{2.5}{1.5}{\aru};\darr{2.5}{2.0}{\aru};\darr{2.5}{2.5}{\aru};\darr{2.5}{3.0}{\au};\darr{2.5}{3.5}{\alu};\darr{2.5}{4.0}{\al};
\darr{3.0}{0.0}{\al};\darr{3.0}{0.5}{\alu};\darr{3.0}{1.0}{\au};\darr{3.0}{1.5}{\au};\darr{3.0}{2.0}{\au};\darr{3.0}{2.5}{\au};\darr{3.0}{3.0}{\au};\darr{3.0}{3.5}{\alu};\darr{3.0}{4.0}{\al};
\darr{3.5}{0.0}{\al};\darr{3.5}{0.5}{\alu};\darr{3.5}{1.0}{\alu};\darr{3.5}{1.5}{\alu};\darr{3.5}{2.0}{\alu};\darr{3.5}{2.5}{\alu};\darr{3.5}{3.0}{\alu};\darr{3.5}{3.5}{\alu};\darr{3.5}{4.0}{\al};
\darr{4.0}{0.0}{\al};\darr{4.0}{0.5}{\al};\darr{4.0}{1.0}{\al};\darr{4.0}{1.5}{\al};\darr{4.0}{2.0}{\al};\darr{4.0}{2.5}{\al};\darr{4.0}{3.0}{\al};\darr{4.0}{3.5}{\al};\darr{4.0}{4.0}{\al};
\draw[thick, color=red] (1.25, 0.5) ellipse (0.8 and 0.6);

        \end{scope}

        \begin{scope}[xshift= 0 * \xstep, yshift = 0 * \ystep]       
            \begin{scope}[yshift=2mm, xshift=-3mm]
                
                \draw[very thick, ->] (-1.4, 2.0) -> (-1.0, 2.0);
                \node[fill=black, circle, inner sep=0.7mm] at (-1.5, 2) {};

                \node[xshift=3mm] at (-0.7, 2) {$\Rightarrow$};
            \end{scope}

            \draw[step=.5cm,style=help lines] (0, 0) grid (4.5, 4.5);

            \darr{0.0}{0.0}{\al};\darr{0.0}{0.5}{\al};\darr{0.0}{1.0}{\al};\darr{0.0}{1.5}{\al};\darr{0.0}{2.0}{\al};\darr{0.0}{2.5}{\al};\darr{0.0}{3.0}{\al};\darr{0.0}{3.5}{\al};\darr{0.0}{4.0}{\al};
\darr{0.5}{0.0}{\al};\darr{0.5}{0.5}{\alu};\darr{0.5}{1.0}{\alu};\darr{0.5}{1.5}{\alu};\darr{0.5}{2.0}{\alu};\darr{0.5}{2.5}{\alu};\darr{0.5}{3.0}{\alu};\darr{0.5}{3.5}{\alu};\darr{0.5}{4.0}{\ald};
\darr{1.0}{0.0}{\al};\darr{1.0}{0.5}{\alu};\darr{1.0}{1.0}{\au};\darr{1.0}{1.5}{\au};\darr{1.0}{2.0}{\au};\darr{1.0}{2.5}{\au};\darr{1.0}{3.0}{\au};\darr{1.0}{3.5}{\au};\darr{1.0}{4.0}{\ad};
\darr{1.5}{0.0}{\al};\darr{1.5}{0.5}{\alu};\darr{1.5}{1.0}{\au};\darr{1.5}{1.5}{\aru};\darr{1.5}{2.0}{\aru};\darr{1.5}{2.5}{\aru};\darr{1.5}{3.0}{\aru};\darr{1.5}{3.5}{\aru};\darr{1.5}{4.0}{\ard};
\darr{2.0}{0.0}{\al};\darr{2.0}{0.5}{\alu};\darr{2.0}{1.0}{\au};\darr{2.0}{1.5}{\aru};\darr{2.0}{2.0}{\ar};\darr{2.0}{2.5}{\ar};\darr{2.0}{3.0}{\ar};\darr{2.0}{3.5}{\ar};\darr{2.0}{4.0}{\ar};
\darr{2.5}{0.0}{\al};\darr{2.5}{0.5}{\alu};\darr{2.5}{1.0}{\au};\darr{2.5}{1.5}{\aru};\darr{2.5}{2.0}{\aru};\darr{2.5}{2.5}{\aru};\darr{2.5}{3.0}{\aru};\darr{2.5}{3.5}{\aru};\darr{2.5}{4.0}{\aru};
\darr{3.0}{0.0}{\al};\darr{3.0}{0.5}{\alu};\darr{3.0}{1.0}{\au};\darr{3.0}{1.5}{\au};\darr{3.0}{2.0}{\au};\darr{3.0}{2.5}{\au};\darr{3.0}{3.0}{\au};\darr{3.0}{3.5}{\au};\darr{3.0}{4.0}{\au};
\darr{3.5}{0.0}{\al};\darr{3.5}{0.5}{\alu};\darr{3.5}{1.0}{\alu};\darr{3.5}{1.5}{\alu};\darr{3.5}{2.0}{\alu};\darr{3.5}{2.5}{\alu};\darr{3.5}{3.0}{\alu};\darr{3.5}{3.5}{\alu};\darr{3.5}{4.0}{\alu};
\darr{4.0}{0.0}{\al};\darr{4.0}{0.5}{\al};\darr{4.0}{1.0}{\al};\darr{4.0}{1.5}{\al};\darr{4.0}{2.0}{\al};\darr{4.0}{2.5}{\al};\darr{4.0}{3.0}{\al};\darr{4.0}{3.5}{\al};\darr{4.0}{4.0}{\al};
\draw[thick, color=red] (4.0, 1.25) ellipse (0.6 and 0.8);

        \end{scope}

        \begin{scope}[xshift= 1 * \xstep, yshift = 0 * \ystep]       
            \begin{scope}[yshift=2mm, xshift=-3mm]
                \draw[very thick, ->] (-2.0, 2.0) -> (-1.6, 2.0);
                
                \node[fill=black, circle, inner sep=0.7mm] at (-1.5, 2) {};

                \node[xshift=3mm] at (-0.7, 2) {$\Rightarrow$};
            \end{scope}

            \draw[step=.5cm,style=help lines] (0, 0) grid (4.5, 4.5);

            \darr{0.0}{0.0}{\al};\darr{0.0}{0.5}{\al};\darr{0.0}{1.0}{\al};\darr{0.0}{1.5}{\al};\darr{0.0}{2.0}{\al};\darr{0.0}{2.5}{\al};\darr{0.0}{3.0}{\al};\darr{0.0}{3.5}{\al};\darr{0.0}{4.0}{\al};
\darr{0.5}{0.0}{\ald};\darr{0.5}{0.5}{\alu};\darr{0.5}{1.0}{\alu};\darr{0.5}{1.5}{\alu};\darr{0.5}{2.0}{\alu};\darr{0.5}{2.5}{\alu};\darr{0.5}{3.0}{\alu};\darr{0.5}{3.5}{\alu};\darr{0.5}{4.0}{\al};
\darr{1.0}{0.0}{\ad};\darr{1.0}{0.5}{\au};\darr{1.0}{1.0}{\au};\darr{1.0}{1.5}{\au};\darr{1.0}{2.0}{\au};\darr{1.0}{2.5}{\au};\darr{1.0}{3.0}{\au};\darr{1.0}{3.5}{\alu};\darr{1.0}{4.0}{\al};
\darr{1.5}{0.0}{\ard};\darr{1.5}{0.5}{\aru};\darr{1.5}{1.0}{\aru};\darr{1.5}{1.5}{\aru};\darr{1.5}{2.0}{\aru};\darr{1.5}{2.5}{\aru};\darr{1.5}{3.0}{\au};\darr{1.5}{3.5}{\alu};\darr{1.5}{4.0}{\al};
\darr{2.0}{0.0}{\ar};\darr{2.0}{0.5}{\ar};\darr{2.0}{1.0}{\ar};\darr{2.0}{1.5}{\ar};\darr{2.0}{2.0}{\ar};\darr{2.0}{2.5}{\aru};\darr{2.0}{3.0}{\au};\darr{2.0}{3.5}{\alu};\darr{2.0}{4.0}{\al};
\darr{2.5}{0.0}{\aru};\darr{2.5}{0.5}{\aru};\darr{2.5}{1.0}{\aru};\darr{2.5}{1.5}{\aru};\darr{2.5}{2.0}{\aru};\darr{2.5}{2.5}{\aru};\darr{2.5}{3.0}{\au};\darr{2.5}{3.5}{\alu};\darr{2.5}{4.0}{\al};
\darr{3.0}{0.0}{\au};\darr{3.0}{0.5}{\au};\darr{3.0}{1.0}{\au};\darr{3.0}{1.5}{\au};\darr{3.0}{2.0}{\au};\darr{3.0}{2.5}{\au};\darr{3.0}{3.0}{\au};\darr{3.0}{3.5}{\alu};\darr{3.0}{4.0}{\al};
\darr{3.5}{0.0}{\alu};\darr{3.5}{0.5}{\alu};\darr{3.5}{1.0}{\alu};\darr{3.5}{1.5}{\alu};\darr{3.5}{2.0}{\alu};\darr{3.5}{2.5}{\alu};\darr{3.5}{3.0}{\alu};\darr{3.5}{3.5}{\alu};\darr{3.5}{4.0}{\al};
\darr{4.0}{0.0}{\al};\darr{4.0}{0.5}{\al};\darr{4.0}{1.0}{\al};\darr{4.0}{1.5}{\al};\darr{4.0}{2.0}{\al};\darr{4.0}{2.5}{\al};\darr{4.0}{3.0}{\al};\darr{4.0}{3.5}{\al};\darr{4.0}{4.0}{\al};
\draw[thick, color=red] (0.5, 1.25) ellipse (0.6 and 0.8);

        \end{scope}

        \begin{scope}[xshift= 2 * \xstep, yshift = 0 * \ystep]       
            \begin{scope}[yshift=2mm, xshift=-3mm]
                \draw[very thick, ->] (-1.5, 2.5) -> (-1.5, 2.1);
                
                \node[fill=black, circle, inner sep=0.7mm] at (-1.5, 2) {};

                \node[xshift=3mm] at (-0.7, 2) {$\Rightarrow$};
            \end{scope}

            \draw[step=.5cm,style=help lines] (0, 0) grid (4.5, 4.5);

            \darr{0.0}{0.0}{\al};\darr{0.0}{0.5}{\al};\darr{0.0}{1.0}{\al};\darr{0.0}{1.5}{\al};\darr{0.0}{2.0}{\al};\darr{0.0}{2.5}{\al};\darr{0.0}{3.0}{\al};\darr{0.0}{3.5}{\al};\darr{0.0}{4.0}{\al};
\darr{0.5}{0.0}{\al};\darr{0.5}{0.5}{\alu};\darr{0.5}{1.0}{\alu};\darr{0.5}{1.5}{\alu};\darr{0.5}{2.0}{\alu};\darr{0.5}{2.5}{\alu};\darr{0.5}{3.0}{\alu};\darr{0.5}{3.5}{\alu};\darr{0.5}{4.0}{\al};
\darr{1.0}{0.0}{\al};\darr{1.0}{0.5}{\alu};\darr{1.0}{1.0}{\au};\darr{1.0}{1.5}{\au};\darr{1.0}{2.0}{\au};\darr{1.0}{2.5}{\au};\darr{1.0}{3.0}{\au};\darr{1.0}{3.5}{\alu};\darr{1.0}{4.0}{\al};
\darr{1.5}{0.0}{\al};\darr{1.5}{0.5}{\alu};\darr{1.5}{1.0}{\au};\darr{1.5}{1.5}{\aru};\darr{1.5}{2.0}{\aru};\darr{1.5}{2.5}{\aru};\darr{1.5}{3.0}{\au};\darr{1.5}{3.5}{\alu};\darr{1.5}{4.0}{\al};
\darr{2.0}{0.0}{\al};\darr{2.0}{0.5}{\alu};\darr{2.0}{1.0}{\au};\darr{2.0}{1.5}{\aru};\darr{2.0}{2.0}{\ar};\darr{2.0}{2.5}{\aru};\darr{2.0}{3.0}{\au};\darr{2.0}{3.5}{\alu};\darr{2.0}{4.0}{\al};
\darr{2.5}{0.0}{\al};\darr{2.5}{0.5}{\alu};\darr{2.5}{1.0}{\au};\darr{2.5}{1.5}{\aru};\darr{2.5}{2.0}{\ar};\darr{2.5}{2.5}{\aru};\darr{2.5}{3.0}{\au};\darr{2.5}{3.5}{\alu};\darr{2.5}{4.0}{\al};
\darr{3.0}{0.0}{\al};\darr{3.0}{0.5}{\alu};\darr{3.0}{1.0}{\au};\darr{3.0}{1.5}{\aru};\darr{3.0}{2.0}{\ar};\darr{3.0}{2.5}{\aru};\darr{3.0}{3.0}{\au};\darr{3.0}{3.5}{\alu};\darr{3.0}{4.0}{\al};
\darr{3.5}{0.0}{\al};\darr{3.5}{0.5}{\alu};\darr{3.5}{1.0}{\au};\darr{3.5}{1.5}{\aru};\darr{3.5}{2.0}{\ar};\darr{3.5}{2.5}{\aru};\darr{3.5}{3.0}{\au};\darr{3.5}{3.5}{\alu};\darr{3.5}{4.0}{\al};
\darr{4.0}{0.0}{\al};\darr{4.0}{0.5}{\ald};\darr{4.0}{1.0}{\ad};\darr{4.0}{1.5}{\ard};\darr{4.0}{2.0}{\ar};\darr{4.0}{2.5}{\aru};\darr{4.0}{3.0}{\au};\darr{4.0}{3.5}{\alu};\darr{4.0}{4.0}{\al};
\draw[thick, color=red] (1.25, 4.0) ellipse (0.8 and 0.6);

        \end{scope}

    \end{tikzpicture}
}
\vspace{0.3cm}

%% file: pics/a-example.tex
\vspace{0.3cm}
{
    \tikzstyle{arr}=[xshift=7, yshift=7]
    \tikzstyle{lw}=[very thick, line width=0.5mm]
    \newcommand{\darr}[3]{\node[arr] at (#2, #1) {$#3$}}
    \newcommand{\darrl}[9]{\darr{#1}{0}{#2};\darr{#1}{1}{#3};\darr{#1}{2}{#4};\darr{#1}{3}{#5};\darr{#1}{4}{#6};\darr{#1}{5}{#7};\darr{#1}{6}{#8};\darr{#1}{7}{#9};}
    
    \newcommand{\xstep}{4.5cm}
    \newcommand{\ystep}{4.5cm}

    \begin{tikzpicture}[scale=1]

\draw[step=.5cm,style=help lines] (0, 0) grid (23 * 0.5, 20 * 0.5);

\draw[lw] (2.0, 0.5) -- (11.0, 0.5);

\draw[lw] (2.0, 5.0) -- (11.0, 5.0);

\draw[lw] (2.0, 9.5) -- (11.0, 9.5);

\draw[lw] (2.0, 0.5) -- (2.0, 9.5);

\draw[lw] (6.5, 0.5) -- (6.5, 9.5);

\draw[lw] (11.0, 0.5) -- (11.0, 9.5);

\draw[lw] (0, 0.5) -- (2, 0.5);

\draw[lw] (0, 5) -- (2, 5);

            \begin{scope}[yshift=4.5cm, xshift=-3cm]
                \node[fill=black, circle, inner sep=0.7mm] at (0.0, 0.5) {};
\draw[very thick, ->] (0.1, 0.5) -> (0.4, 0.5);\node[fill=black, circle, inner sep=0.7mm] at (0.5, 0.5) {};
\draw[very thick, ->] (0.5, 0.4) -> (0.5, 0.1);\node[fill=black, circle, inner sep=0.7mm] at (0.0, 0.0) {};
\draw[very thick, ->] (0.0, 0.1) -> (0.0, 0.4);\node[fill=black, circle, inner sep=0.7mm] at (0.5, 0.0) {};

                \node[xshift=3mm] at (1.5, 0.25) {$\Rightarrow$};
            \end{scope}
            
\begin{scope}[xshift=2cm, yshift=0.5cm]

        \begin{scope}[xshift= 0 * \xstep, yshift = 1 * \ystep]   
            \darr{0.0}{0.0}{\al};\darr{0.0}{0.5}{\alu};\darr{0.0}{1.0}{\au};\darr{0.0}{1.5}{\aru};\darr{0.0}{2.0}{\ar};\darr{0.0}{2.5}{\ard};\darr{0.0}{3.0}{\ad};\darr{0.0}{3.5}{\ald};\darr{0.0}{4.0}{\al};
\darr{0.5}{0.0}{\al};\darr{0.5}{0.5}{\alu};\darr{0.5}{1.0}{\au};\darr{0.5}{1.5}{\aru};\darr{0.5}{2.0}{\ar};\darr{0.5}{2.5}{\ard};\darr{0.5}{3.0}{\ad};\darr{0.5}{3.5}{\ald};\darr{0.5}{4.0}{\ald};
\darr{1.0}{0.0}{\al};\darr{1.0}{0.5}{\alu};\darr{1.0}{1.0}{\au};\darr{1.0}{1.5}{\aru};\darr{1.0}{2.0}{\ar};\darr{1.0}{2.5}{\ard};\darr{1.0}{3.0}{\ad};\darr{1.0}{3.5}{\ad};\darr{1.0}{4.0}{\ad};
\darr{1.5}{0.0}{\al};\darr{1.5}{0.5}{\alu};\darr{1.5}{1.0}{\au};\darr{1.5}{1.5}{\aru};\darr{1.5}{2.0}{\ar};\darr{1.5}{2.5}{\ard};\darr{1.5}{3.0}{\ard};\darr{1.5}{3.5}{\ard};\darr{1.5}{4.0}{\ard};
\darr{2.0}{0.0}{\al};\darr{2.0}{0.5}{\alu};\darr{2.0}{1.0}{\au};\darr{2.0}{1.5}{\aru};\darr{2.0}{2.0}{\ar};\darr{2.0}{2.5}{\ar};\darr{2.0}{3.0}{\ar};\darr{2.0}{3.5}{\ar};\darr{2.0}{4.0}{\ar};
\darr{2.5}{0.0}{\al};\darr{2.5}{0.5}{\alu};\darr{2.5}{1.0}{\au};\darr{2.5}{1.5}{\aru};\darr{2.5}{2.0}{\aru};\darr{2.5}{2.5}{\aru};\darr{2.5}{3.0}{\aru};\darr{2.5}{3.5}{\aru};\darr{2.5}{4.0}{\aru};
\darr{3.0}{0.0}{\al};\darr{3.0}{0.5}{\alu};\darr{3.0}{1.0}{\au};\darr{3.0}{1.5}{\au};\darr{3.0}{2.0}{\au};\darr{3.0}{2.5}{\au};\darr{3.0}{3.0}{\au};\darr{3.0}{3.5}{\au};\darr{3.0}{4.0}{\au};
\darr{3.5}{0.0}{\al};\darr{3.5}{0.5}{\alu};\darr{3.5}{1.0}{\alu};\darr{3.5}{1.5}{\alu};\darr{3.5}{2.0}{\alu};\darr{3.5}{2.5}{\alu};\darr{3.5}{3.0}{\alu};\darr{3.5}{3.5}{\alu};\darr{3.5}{4.0}{\alu};
\darr{4.0}{0.0}{\al};\darr{4.0}{0.5}{\al};\darr{4.0}{1.0}{\al};\darr{4.0}{1.5}{\al};\darr{4.0}{2.0}{\al};\darr{4.0}{2.5}{\al};\darr{4.0}{3.0}{\al};\darr{4.0}{3.5}{\al};\darr{4.0}{4.0}{\al};

        \end{scope}

        \begin{scope}[xshift= 1 * \xstep, yshift = 1 * \ystep]   
            \darr{0.0}{0.0}{\al};\darr{0.0}{0.5}{\ald};\darr{0.0}{1.0}{\ad};\darr{0.0}{1.5}{\ard};\darr{0.0}{2.0}{\ar};\darr{0.0}{2.5}{\aru};\darr{0.0}{3.0}{\au};\darr{0.0}{3.5}{\alu};\darr{0.0}{4.0}{\al};
\darr{0.5}{0.0}{\ald};\darr{0.5}{0.5}{\ald};\darr{0.5}{1.0}{\ad};\darr{0.5}{1.5}{\ard};\darr{0.5}{2.0}{\ar};\darr{0.5}{2.5}{\aru};\darr{0.5}{3.0}{\au};\darr{0.5}{3.5}{\alu};\darr{0.5}{4.0}{\al};
\darr{1.0}{0.0}{\ad};\darr{1.0}{0.5}{\ad};\darr{1.0}{1.0}{\ad};\darr{1.0}{1.5}{\ard};\darr{1.0}{2.0}{\ar};\darr{1.0}{2.5}{\aru};\darr{1.0}{3.0}{\au};\darr{1.0}{3.5}{\alu};\darr{1.0}{4.0}{\al};
\darr{1.5}{0.0}{\ard};\darr{1.5}{0.5}{\ard};\darr{1.5}{1.0}{\ard};\darr{1.5}{1.5}{\ard};\darr{1.5}{2.0}{\ar};\darr{1.5}{2.5}{\aru};\darr{1.5}{3.0}{\au};\darr{1.5}{3.5}{\alu};\darr{1.5}{4.0}{\al};
\darr{2.0}{0.0}{\ar};\darr{2.0}{0.5}{\ar};\darr{2.0}{1.0}{\ar};\darr{2.0}{1.5}{\ar};\darr{2.0}{2.0}{\ar};\darr{2.0}{2.5}{\aru};\darr{2.0}{3.0}{\au};\darr{2.0}{3.5}{\alu};\darr{2.0}{4.0}{\al};
\darr{2.5}{0.0}{\aru};\darr{2.5}{0.5}{\aru};\darr{2.5}{1.0}{\aru};\darr{2.5}{1.5}{\aru};\darr{2.5}{2.0}{\aru};\darr{2.5}{2.5}{\aru};\darr{2.5}{3.0}{\au};\darr{2.5}{3.5}{\alu};\darr{2.5}{4.0}{\al};
\darr{3.0}{0.0}{\au};\darr{3.0}{0.5}{\au};\darr{3.0}{1.0}{\au};\darr{3.0}{1.5}{\au};\darr{3.0}{2.0}{\au};\darr{3.0}{2.5}{\au};\darr{3.0}{3.0}{\au};\darr{3.0}{3.5}{\alu};\darr{3.0}{4.0}{\al};
\darr{3.5}{0.0}{\alu};\darr{3.5}{0.5}{\alu};\darr{3.5}{1.0}{\alu};\darr{3.5}{1.5}{\alu};\darr{3.5}{2.0}{\alu};\darr{3.5}{2.5}{\alu};\darr{3.5}{3.0}{\alu};\darr{3.5}{3.5}{\alu};\darr{3.5}{4.0}{\al};
\darr{4.0}{0.0}{\al};\darr{4.0}{0.5}{\al};\darr{4.0}{1.0}{\al};\darr{4.0}{1.5}{\al};\darr{4.0}{2.0}{\al};\darr{4.0}{2.5}{\al};\darr{4.0}{3.0}{\al};\darr{4.0}{3.5}{\al};\darr{4.0}{4.0}{\al};

        \end{scope}

        \begin{scope}[xshift= 0 * \xstep, yshift = 0 * \ystep]   
            \darr{0.0}{0.0}{\al};\darr{0.0}{0.5}{\al};\darr{0.0}{1.0}{\al};\darr{0.0}{1.5}{\al};\darr{0.0}{2.0}{\al};\darr{0.0}{2.5}{\al};\darr{0.0}{3.0}{\al};\darr{0.0}{3.5}{\al};\darr{0.0}{4.0}{\al};
\darr{0.5}{0.0}{\ald};\darr{0.5}{0.5}{\ald};\darr{0.5}{1.0}{\ald};\darr{0.5}{1.5}{\ald};\darr{0.5}{2.0}{\ald};\darr{0.5}{2.5}{\ald};\darr{0.5}{3.0}{\ald};\darr{0.5}{3.5}{\ald};\darr{0.5}{4.0}{\al};
\darr{1.0}{0.0}{\ad};\darr{1.0}{0.5}{\ad};\darr{1.0}{1.0}{\ad};\darr{1.0}{1.5}{\ad};\darr{1.0}{2.0}{\ad};\darr{1.0}{2.5}{\ad};\darr{1.0}{3.0}{\ad};\darr{1.0}{3.5}{\ald};\darr{1.0}{4.0}{\al};
\darr{1.5}{0.0}{\ard};\darr{1.5}{0.5}{\ard};\darr{1.5}{1.0}{\ard};\darr{1.5}{1.5}{\ard};\darr{1.5}{2.0}{\ard};\darr{1.5}{2.5}{\ard};\darr{1.5}{3.0}{\ad};\darr{1.5}{3.5}{\ald};\darr{1.5}{4.0}{\al};
\darr{2.0}{0.0}{\ar};\darr{2.0}{0.5}{\ar};\darr{2.0}{1.0}{\ar};\darr{2.0}{1.5}{\ar};\darr{2.0}{2.0}{\ar};\darr{2.0}{2.5}{\ard};\darr{2.0}{3.0}{\ad};\darr{2.0}{3.5}{\ald};\darr{2.0}{4.0}{\al};
\darr{2.5}{0.0}{\aru};\darr{2.5}{0.5}{\aru};\darr{2.5}{1.0}{\aru};\darr{2.5}{1.5}{\aru};\darr{2.5}{2.0}{\ar};\darr{2.5}{2.5}{\ard};\darr{2.5}{3.0}{\ad};\darr{2.5}{3.5}{\ald};\darr{2.5}{4.0}{\al};
\darr{3.0}{0.0}{\au};\darr{3.0}{0.5}{\au};\darr{3.0}{1.0}{\au};\darr{3.0}{1.5}{\aru};\darr{3.0}{2.0}{\ar};\darr{3.0}{2.5}{\ard};\darr{3.0}{3.0}{\ad};\darr{3.0}{3.5}{\ald};\darr{3.0}{4.0}{\al};
\darr{3.5}{0.0}{\alu};\darr{3.5}{0.5}{\alu};\darr{3.5}{1.0}{\au};\darr{3.5}{1.5}{\aru};\darr{3.5}{2.0}{\ar};\darr{3.5}{2.5}{\ard};\darr{3.5}{3.0}{\ad};\darr{3.5}{3.5}{\ald};\darr{3.5}{4.0}{\al};
\darr{4.0}{0.0}{\al};\darr{4.0}{0.5}{\alu};\darr{4.0}{1.0}{\au};\darr{4.0}{1.5}{\aru};\darr{4.0}{2.0}{\ar};\darr{4.0}{2.5}{\ard};\darr{4.0}{3.0}{\ad};\darr{4.0}{3.5}{\ald};\darr{4.0}{4.0}{\al};

        \end{scope}

        \begin{scope}[xshift= 1 * \xstep, yshift = 0 * \ystep]   
            \darr{0.0}{0.0}{\al};\darr{0.0}{0.5}{\al};\darr{0.0}{1.0}{\al};\darr{0.0}{1.5}{\al};\darr{0.0}{2.0}{\al};\darr{0.0}{2.5}{\al};\darr{0.0}{3.0}{\al};\darr{0.0}{3.5}{\al};\darr{0.0}{4.0}{\al};
\darr{0.5}{0.0}{\al};\darr{0.5}{0.5}{\alu};\darr{0.5}{1.0}{\alu};\darr{0.5}{1.5}{\alu};\darr{0.5}{2.0}{\alu};\darr{0.5}{2.5}{\alu};\darr{0.5}{3.0}{\alu};\darr{0.5}{3.5}{\alu};\darr{0.5}{4.0}{\al};
\darr{1.0}{0.0}{\al};\darr{1.0}{0.5}{\alu};\darr{1.0}{1.0}{\au};\darr{1.0}{1.5}{\au};\darr{1.0}{2.0}{\au};\darr{1.0}{2.5}{\au};\darr{1.0}{3.0}{\au};\darr{1.0}{3.5}{\alu};\darr{1.0}{4.0}{\al};
\darr{1.5}{0.0}{\al};\darr{1.5}{0.5}{\alu};\darr{1.5}{1.0}{\au};\darr{1.5}{1.5}{\aru};\darr{1.5}{2.0}{\aru};\darr{1.5}{2.5}{\aru};\darr{1.5}{3.0}{\au};\darr{1.5}{3.5}{\alu};\darr{1.5}{4.0}{\al};
\darr{2.0}{0.0}{\al};\darr{2.0}{0.5}{\alu};\darr{2.0}{1.0}{\au};\darr{2.0}{1.5}{\aru};\darr{2.0}{2.0}{\ar};\darr{2.0}{2.5}{\aru};\darr{2.0}{3.0}{\au};\darr{2.0}{3.5}{\alu};\darr{2.0}{4.0}{\al};
\darr{2.5}{0.0}{\al};\darr{2.5}{0.5}{\alu};\darr{2.5}{1.0}{\au};\darr{2.5}{1.5}{\aru};\darr{2.5}{2.0}{\ar};\darr{2.5}{2.5}{\aru};\darr{2.5}{3.0}{\au};\darr{2.5}{3.5}{\alu};\darr{2.5}{4.0}{\al};
\darr{3.0}{0.0}{\al};\darr{3.0}{0.5}{\alu};\darr{3.0}{1.0}{\au};\darr{3.0}{1.5}{\aru};\darr{3.0}{2.0}{\ar};\darr{3.0}{2.5}{\aru};\darr{3.0}{3.0}{\au};\darr{3.0}{3.5}{\alu};\darr{3.0}{4.0}{\al};
\darr{3.5}{0.0}{\al};\darr{3.5}{0.5}{\alu};\darr{3.5}{1.0}{\au};\darr{3.5}{1.5}{\aru};\darr{3.5}{2.0}{\ar};\darr{3.5}{2.5}{\aru};\darr{3.5}{3.0}{\au};\darr{3.5}{3.5}{\alu};\darr{3.5}{4.0}{\al};
\darr{4.0}{0.0}{\al};\darr{4.0}{0.5}{\ald};\darr{4.0}{1.0}{\ad};\darr{4.0}{1.5}{\ard};\darr{4.0}{2.0}{\ar};\darr{4.0}{2.5}{\aru};\darr{4.0}{3.0}{\au};\darr{4.0}{3.5}{\alu};\darr{4.0}{4.0}{\al};
\draw[thick, color=red] (1.25, 4.0) ellipse (0.8 and 0.6);

        \end{scope}

\end{scope}

        \begin{scope}[xshift= 0cm, yshift = 0cm]   
            \darr{0.0}{0.0}{\aru};\darr{0.0}{0.5}{\au};\darr{0.0}{1.0}{\alu};\darr{0.0}{1.5}{\alu};\darr{0.0}{2.0}{\alu};\darr{0.0}{2.5}{\alu};\darr{0.0}{3.0}{\alu};\darr{0.0}{3.5}{\alu};\darr{0.0}{4.0}{\alu};\darr{0.0}{4.5}{\alu};\darr{0.0}{5.0}{\alu};\darr{0.0}{5.5}{\alu};\darr{0.0}{6.0}{\alu};\darr{0.0}{6.5}{\alu};\darr{0.0}{7.0}{\alu};\darr{0.0}{7.5}{\alu};\darr{0.0}{8.0}{\alu};\darr{0.0}{8.5}{\alu};\darr{0.0}{9.0}{\alu};\darr{0.0}{9.5}{\alu};\darr{0.0}{10.0}{\alu};\darr{0.0}{10.5}{\alu};\darr{0.0}{11.0}{\alu};
\darr{0.5}{0.0}{\aru};\darr{0.5}{0.5}{\au};\darr{0.5}{1.0}{\alu};\darr{0.5}{1.5}{\alu};\darr{0.5}{11.0}{\al};
\darr{1.0}{0.0}{\aru};\darr{1.0}{0.5}{\au};\darr{1.0}{1.0}{\alu};\darr{1.0}{1.5}{\alu};\darr{1.0}{11.0}{\al};
\darr{1.5}{0.0}{\aru};\darr{1.5}{0.5}{\au};\darr{1.5}{1.0}{\au};\darr{1.5}{1.5}{\au};\darr{1.5}{11.0}{\al};
\darr{2.0}{0.0}{\aru};\darr{2.0}{0.5}{\aru};\darr{2.0}{1.0}{\aru};\darr{2.0}{1.5}{\aru};\darr{2.0}{11.0}{\al};
\darr{2.5}{0.0}{\ar};\darr{2.5}{0.5}{\ar};\darr{2.5}{1.0}{\ar};\darr{2.5}{1.5}{\ar};\darr{2.5}{11.0}{\al};
\darr{3.0}{0.0}{\ard};\darr{3.0}{0.5}{\ard};\darr{3.0}{1.0}{\ard};\darr{3.0}{1.5}{\ard};\darr{3.0}{11.0}{\al};
\darr{3.5}{0.0}{\ard};\darr{3.5}{0.5}{\ad};\darr{3.5}{1.0}{\ald};\darr{3.5}{1.5}{\ad};\darr{3.5}{11.0}{\al};
\darr{4.0}{0.0}{\ard};\darr{4.0}{0.5}{\ad};\darr{4.0}{1.0}{\ald};\darr{4.0}{1.5}{\ald};\darr{4.0}{11.0}{\al};
\darr{4.5}{0.0}{\ard};\darr{4.5}{0.5}{\ad};\darr{4.5}{1.0}{\ald};\darr{4.5}{1.5}{\al};\darr{4.5}{11.0}{\al};
\darr{5.0}{0.0}{\ard};\darr{5.0}{0.5}{\ad};\darr{5.0}{1.0}{\ald};\darr{5.0}{1.5}{\al};\darr{5.0}{11.0}{\al};
\darr{5.5}{0.0}{\ard};\darr{5.5}{0.5}{\ad};\darr{5.5}{1.0}{\ald};\darr{5.5}{1.5}{\al};\darr{5.5}{11.0}{\al};
\darr{6.0}{0.0}{\ard};\darr{6.0}{0.5}{\ad};\darr{6.0}{1.0}{\ald};\darr{6.0}{1.5}{\al};\darr{6.0}{11.0}{\al};
\darr{6.5}{0.0}{\ard};\darr{6.5}{0.5}{\ad};\darr{6.5}{1.0}{\ald};\darr{6.5}{1.5}{\al};\darr{6.5}{11.0}{\al};
\darr{7.0}{0.0}{\ard};\darr{7.0}{0.5}{\ad};\darr{7.0}{1.0}{\ald};\darr{7.0}{1.5}{\al};\darr{7.0}{11.0}{\al};
\darr{7.5}{0.0}{\ard};\darr{7.5}{0.5}{\ad};\darr{7.5}{1.0}{\ald};\darr{7.5}{1.5}{\al};\darr{7.5}{11.0}{\al};
\darr{8.0}{0.0}{\ard};\darr{8.0}{0.5}{\ad};\darr{8.0}{1.0}{\ald};\darr{8.0}{1.5}{\al};\darr{8.0}{11.0}{\al};
\darr{8.5}{0.0}{\ard};\darr{8.5}{0.5}{\ad};\darr{8.5}{1.0}{\ald};\darr{8.5}{1.5}{\al};\darr{8.5}{11.0}{\al};
\darr{9.0}{0.0}{\ard};\darr{9.0}{0.5}{\ad};\darr{9.0}{1.0}{\ald};\darr{9.0}{1.5}{\al};\darr{9.0}{11.0}{\al};
\darr{9.5}{0.0}{\ard};\darr{9.5}{0.5}{\ad};\darr{9.5}{1.0}{\ald};\darr{9.5}{1.5}{\ald};\darr{9.5}{2.0}{\ald};\darr{9.5}{2.5}{\ald};\darr{9.5}{3.0}{\ald};\darr{9.5}{3.5}{\ald};\darr{9.5}{4.0}{\ald};\darr{9.5}{4.5}{\ald};\darr{9.5}{5.0}{\ald};\darr{9.5}{5.5}{\ald};\darr{9.5}{6.0}{\ald};\darr{9.5}{6.5}{\ald};\darr{9.5}{7.0}{\ald};\darr{9.5}{7.5}{\ald};\darr{9.5}{8.0}{\ald};\darr{9.5}{8.5}{\ald};\darr{9.5}{9.0}{\ald};\darr{9.5}{9.5}{\ald};\darr{9.5}{10.0}{\ald};\darr{9.5}{10.5}{\ald};\darr{9.5}{11.0}{\ald};

        \end{scope}

    \end{tikzpicture}
}
\vspace{0.3cm}

%% file: parts/appendix.tex
\section{Proofs from Section~\ref{sectarrows}}
\label{appendixarrow}

The unsatisfiability of $\psi_n$ follows from Brower's theorem, but we give an alternative proof that will be used
in proving upper bound on the size of a contradiction search tree.

Let us consider arrows in two adjacent by edge cells and suppose the first arrow should be rotate on $x$ degree in the clockwise way to get the second one. We also suppose that $x$ is minimal in absolute value angle. The rotation between two arrows is signed $x$.

Now let us consider a closed path on the board, that goes through cells with common edges. 
We go along the path and calculate the sum of rotations between neughbouring arrows. Since the path returns back to the initial cell, the total rotation is divisible by $360^\circ$. 

\begin{lemma} \label{boundary}
Consider a rectangle, with arrows in its cells, such that the total rotation on the closed path going along the 
boundary of the rectangle is nonzero. Then there are two cells with the common edge such that the angle between arrows in them
is more than $45^\circ$.
\end{lemma}

\begin{proof}
We give the proof by induction on the size of rectangle. The base of the induction is $2\times 2$ rectangle.
Consider rectangle $R$ with total rotation along the boundary $S$. We divide rectangle $R$ on two parts as it is shown on Figure~\ref{arrow-rotate}.
Let $S_1$ and $S_2$ be total rotation along the boundary of smaller rectangles $R_1$ and $R_2$. Note that
   $S_1 + S_2 = S$. Hence for one of the smaller rectangle the total rotation is nonzero.
   
    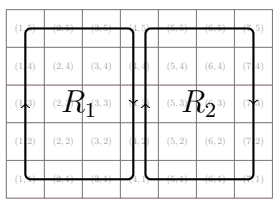
\begin{figure}[h]
        \center{\input{pics/a-division.tex}}
        \caption{Division of $R$ into $R_1$ and $R_2$. \label{arrow-rotate}}
    \end{figure}
    
\end{proof}

\begin{theorem}
    $\CSP$ formula $\psi_n$ is unsatisfiable. 
\end{theorem}

\begin{proof}
We prove that the total rotation along the boundary of $n \times n$ sqaure is $360^\circ$. 
Consider upper left cell $A$ and lower right cell $B$. 
Let the total rotation along the upper path from $A$ to $B$ is $l + 360^\circ k$, where $0 \le l < 360^\circ, k \in \mathbb{Z}$. 
An arrow in the cell $A$ belongs to set $\{\rightarrow, \searrow,
    \downarrow\}$ and one in the cell $B$ belongs to set $\{\leftarrow, \nwarrow, \uparrow\}$, therefore $l \neq 0$. 
Note that if $k \neq 0$, then on the path from $A$ to $B$ there exists $\aru$ that contradicts to the boundary constraints. 
So we have that the total rotation from $A$ to $B$ is positive and less than $360^\circ$. 
Similarly the total rotation along lower path from $B$ to $A$ is positive and less than $360^\circ$. 
Thus the total rotation along the closed path is positive and less than $720^\circ$. But the total rotation should be 
divisible by $360^\circ$, therefore it equals $360^\circ$.

The Theorem follows from Lemma~\ref{boundary}.
\end{proof}

\begin{corollary}
    For the $\CSP$ formula $\psi_n$ there exists a contradiction search tree of depth $O(n)$ and therefore of size $2^{O(n)}$.
\end{corollary}
\begin{proof}
Make a request to all boundary cells. After it we split a square into two approximate equal parts by a column.
Since the total rotation  on the boundary of two parts is nonzero, then by Lemma~~\ref{boundary} one of these parts 
has nonzero rotation and therefore contains a contradiction. We split the contradictory part by a row and 
reduce the problem of a contradiction search to a rectangle of size at most $(\frac{n}{2} + 1)\times (\frac{n}{2} + 1)$ with known 
arrows on the boundary. So we make approximately $1.5 n$ requests and reduce the problem to two times smaller problem. 
The depth of the resulting tree is $O(n)$.
\end{proof}

%% file: pics/a-division.tex
\vspace{0.3cm}
\begin{center}


  \begin{tikzpicture}[scale=1]
    \draw[step=0.5cm, draw=gray, thin,fill=red] (0, 0) grid (3.5, 2.5);
    \foreach \x in {1,...,7} {
      \foreach \y in {1,...,5} {
        \node[scale=0.3,color=gray!80!white] at (-0.25 + \x * 0.5, -0.25 + \y * 0.5) {$(\x, \y)$};
      }
    } 
    \draw[shift={(0.25, 0.25)},rounded corners=2pt,thick] (0, 0) rectangle node {$R_1$} (1.42, 2.0) ;
    \draw[shift={(1.75, 0.25)},rounded corners=2pt,thick] (0.08, 0) rectangle node {$R_2$}  (1.5, 2.0);
    \draw[->] (0.25, 1.25) to (0.25, 1.26);
    \draw[->] (0.25 + 1.42, 1.25) to (0.25 + 1.42, 1.24);
    \draw[->] (0.25 + 1.5 + 0.08, 1.25) to (0.25 + 1.5 + 0.08, 1.26);
    \draw[->] (0.25 + 1.5 + 0.08 + 1.42, 1.25) to (0.25 + 1.5 + 0.08 + 1.42, 1.24);
  \end{tikzpicture}
\end{center}
\vspace{0.3cm}